\documentclass[12pt,centertags,reqno]{amsart}

\usepackage[foot]{amsaddr}
\usepackage{latexsym}
\usepackage[english]{babel}
\usepackage[T1]{fontenc}
\usepackage[utf8]{inputenc}
\usepackage{natbib}
\usepackage{amssymb}
\usepackage{fancyhdr}
\usepackage{url}
\usepackage{hyperref}
\usepackage{verbatim}
\usepackage{leftidx}
\usepackage{color,graphicx}

\usepackage{mathrsfs, mathtools}
\usepackage{stmaryrd}

\usepackage{marvosym}
\mathtoolsset{showonlyrefs}

\usepackage{soul}

\usepackage{enumerate} 
\usepackage{enumitem} 

\numberwithin{equation}{section} \makeatletter
\renewcommand{\subsection}{\@startsection
{subsection}{2}{0mm}{\baselineskip}{-0.25cm}
{\normalfont\normalsize\bf}} \makeatother

\newtheorem{theorem}{Theorem}[section]
\newtheorem{lemma}[theorem]{Lemma}
\newtheorem{corollary}[theorem]{Corollary}
\newtheorem{definition}[theorem]{Definition}
\newtheorem{remark}[theorem]{Remark}
\newtheorem{proposition}[theorem]{Proposition}

\newtheorem{ass}[theorem]{Assumption}

\def \A {\mathcal A}

\def \D {\mathcal D}

\def \F {\mathcal F}

\def \L {\mathcal L}

\def \P {\mathbf P}

\def \R {\mathbb R}
\def \bF {\mathbb F}

\def \bN {\mathbb N}

\newcommand{\pd}[2]{\dfrac{\partial#1}{\partial#2}}
\newcommand{\pds}[2]{\dfrac{\partial{^{2}}#1}{\partial{#2}^{2}}}
\newcommand{\pdsm}[3]{\dfrac{\partial{^{2}}#1}{\partial{#2} \partial{#3}}}

\newcommand{\ud}{\mathrm d}
\newcommand{\ds}{\displaystyle}
\newcommand{\esp}[2][\mathbb E] {#1\left[#2\right]}

\newcommand{\aC}[1]{{\color{red} #1}}

\begin{document}
	
	\author[K.~Colaneri]{Katia Colaneri}\address{Katia Colaneri, Department of Economics and Finance, University of Rome Tor Vergata, Via Columbia 2, 00133 Rome, Italy.}\email{katia.colaneri@uniroma2.it}

	\author[A.~Cretarola]{Alessandra Cretarola}\address{Alessandra Cretarola, Department of Mathematics and Computer Science, University of Perugia, Via L. Vanvitelli, 1, 06123 Perugia, Italy.}\email{alessandra.cretarola@unipg.it}

	\author[B.~Salterini]{Benedetta Salterini}\address{Benedetta Salterini, Department of Mathematics and Computer Science (DIMAI), University of Firenze, Viale Morgagni, 67/a - 50134 Firenze, Italy.}\email{benedetta.salterini@unifi.it}

\title[]{Optimal investment and reinsurance
under exponential forward preferences}

	
\begin{abstract} \begin{center}
We study the optimal investment and proportional reinsurance problem of an insurance company, whose investment preferences are described via a forward dynamic utility of exponential type in a stochastic factor model allowing for a possible dependence between the financial and insurance markets.
Specifically,
we assume that the asset price process dynamics and the claim arrival intensity are both affected by a common stochastic process and we account for a possible environmental contagion effect through the non-zero correlation parameter between the underlying Brownian motions driving the asset price process and the stochastic factor dynamics.
By stochastic control techniques, we construct a forward dynamic exponential utility, and we characterize the optimal investment and reinsurance strategy. Moreover, we investigate in detail the zero-volatility case and provide a comparison analysis with classical results in an analogous setting under backward utility preferences. We also discuss an extension of the conditional certainty equivalent. Finally, we perform a numerical analysis to highlight some features of the optimal strategy.
\end{center}\end{abstract}

\maketitle



%





\section{Introduction}
	
In this paper we investigate the problem of an insurance company which buys proportional reinsurance to protect against the risk of incoming claims and it is also allowed to invest part of its wealth in the market. However, we do not assume that the preferences of the insurance company are represented via traditional utility function but we consider
{\em forward dynamic utilities}. Intuitively, a forward dynamic utility represents individual preferences of an agent, possibly changing over time. The idea of using a utility to define preference for advancing the timing of future satisfaction or impatience is well known in the economic literature. An agent may dynamically adjust her preferences consistently with the information revealed over time and her impatience might be compensated for by the opportunities given to her, if they can be exploited in full according to her choices.

The notion of forward dynamic utility and the forward approach to optimal investment decision problems have been introduced and developed in \citep{MZ,musiela2008optimal,MZportchoice}, to overcome a few limitations of the traditional backward preference approach. Classical literature on portfolio optimization under backward utilities is based on the assumption that a utility is exogenously chosen to hold at a future date (no earlier than the end of the investment horizon) and employed to make investment decisions for today; this means that, when entering the market, agents
prescribe their risk profile at the horizon time and therefore cannot adapt it to changes in market conditions or update risk preferences.
In addition, the investment horizon is fixed, and the portfolio is derived with respect to this reference date. In the forward approach, instead of pre-specifying the utility function to be valid at some future time and identifying a trading horizon, the agent gives today's preferences and the utility is generated forward in time, that is, it naturally moves in the same direction of the market.
The agent chooses the optimal strategy to maximize the expected forward utility of her wealth, {\em at any future time} $t \ge 0$. The main consequence of this approach is that a forward utility allocates the
same value, in terms of utility of wealth, to the optimal investment over any
investment horizon. This yields to a stochastic control problem where the 
value process is determined so that it enjoys the semimartingale property for any admissible strategy and it is a martingale for the optimal strategy. This reflects the natural idea that any sub-optimal strategy is under-performing, and that the expected performance of an optimal strategy at any future time is as good as today. In particular, the martingale property allows us to derive a PDE that characterizes the value function, and hence to 
solve explicitly the problem.\\
In this paper we make use of the forward approach to study an optimal investment-reinsurance problem. This type of problems are widely studied in the actuarial literature under traditional backward utilities preferences, see e.g. \citep{IRPAU,liu2009optimal,gu2017optimal,bec,BRASCHMIDLI,cao2020optimal, ceci2022optimal}. So far, to the best of our knowledge, a first contribution of the forward dynamic approach in non-life insurance is given by \citep{colaneri2021optimal}, where a zero volatility
forward dynamic exponential utility in a regime-switching market model is considered. Nevertheless, we mention that the forward utility approach has been applied recently in the life-insurance framework in \citep{chong2019pricing}, where the evaluation problem of equity-linked life insurance contracts is investigated.\\
The main novelty of our paper is to consider an investment-reinsurance problem of an insurance company whose utility preferences are described by a non-zero volatility forward dynamic exponential utility, possibly depending on time and some stochastic factor, which is {\em correlated with the financial market}. This extends the classical non-life insurance setting by introducing a class of more flexible utility preferences, in a stochastic factor model.\\
The insurance company can invest its wealth in a market where a risky asset and a riskless asset are traded and can buy a proportional reinsurance. \\
Our model is set to encompass a desirable characteristic of hybrid markets that is {\em mutual dependence} between the financial and insurance markets via two different kinds of dependence. First, we assume that
the asset price process dynamics and the claim arrival intensity are both affected by a common stochastic process, which may represent for instance a social/environmental/cultural factor; second, we allow for a possible environmental contagion effect through the non-zero correlation parameter between the underlying Brownian motions driving the asset price process and the stochastic factor dynamics. \\
By the martingale property and a PDE characterization, we obtain an analytic construction of a forward dynamic exponential utility, see Theorem \ref{familyfeu},
and provide the optimal reinsurance-investment strategy, see Proposition \ref{prop:optimal}. For the zero-volatility case, we discuss the differences and similarities with the optimal strategies and the value function under classical (dynamic) backward utilities in Section \ref{sec:comparison}. An important feature of the optimal strategy in the backward case is that it presents the classical two components, a myopic part and an additional risk adjustment. Such adjustment does not apply in the forward case: the investment strategy is purely myopic and this is due to the fact that the risk arising from incompleteness is accounted in the utility function rather than delegated the insurance company actions. Finally, we analyze a dynamic version of the conditional certainty equivalent for forward utility preferences and compare it with the backward utility setting.\\
The paper is organized as follows. Section \ref{sec:comb_model} introduces the mathematical framework for the interdependent insurance and financial market model. In Section \ref{sec:results} we describe the optimization problem, construct the forward dynamic exponential utility and characterize the optimal investment and proportional reinsurance strategy.
In Section \ref{sec:zerovol}, we discuss in detail the zero-volatility case and provide a comparison analysis with classical results on optimal investment and reinsurance problems via backward utility preferences.  The conditional certainty equivalent and its properties under forward and backward utility preferences can be found in Section \ref{sec:CCE}.
Some numerical experiments can be found in Section \ref{sec:numerics}. Finally, some assumptions and most technical proofs are collected in Appendix \ref{sec:tech_res}.

\section{The interdependent insurance-financial market model}\label{sec:comb_model}

We consider a complete probability space $(\Omega, \F, \P)$ endowed with a filtration   $\bF=\{\F_t,\ t \ge 0\}$, satisfying the usual conditions of completeness and right continuity. All processes introduced below are assumed to be $\bF$-adapted.

Let the process $Y=\{Y_{t}, \ t \ge 0\}$ be the solution of the stochastic differential equation (SDE)
\begin{equation}
\label{y}
\ud Y_{t}=\alpha(t,Y_{t}) \ud t + \beta(t,Y_{t}) \ud W_{t}^{Y}, \quad  Y_{0}=y_0 \in \R,
\end{equation}
where, $W^{Y}=\{W^{Y}_{t}, \ t \ge 0\}$ is a standard Brownian motion on $(\Omega, \F, \P; \bF)$, $\alpha: [0,+ \infty) \times \R \to \R$ and $\beta: [0,+ \infty) \times \R \to \R$ are two measurable functions. We assume that there exists a unique strong solution to the SDE \eqref{y} such that \begin{equation}
	\label{Ycoeff}
	\esp{\int_0^t |\alpha(s,Y_s)| \ud s}< \infty,
\qquad \esp{\int_0^t \beta^2(s,Y_s) \ud s}< \infty,
\end{equation}
for every $t\ge 0$ (this holds true under classical sufficient conditions on model coefficients, see e.g. \citep{GIHMSKOR}).
We interpret the stochastic process $Y$ as an index that accounts for an environmental/social/cultural factor and which may affect both the financial and the insurance markets.

We now introduce the actuarial and the financial market model. Starting with the insurance framework,  we represent the losses of the insurance company following a standard construction of the claim amount process, see for instance \citep[Chapter 2]{grandell1991}.
Define the process $\Lambda=\{\Lambda_t,\ t \ge 0\}$ as
\[
\Lambda_t=\int_0^t\lambda(s, Y_s), \ud s, \quad t \ge 0,
\]
where the function  $\lambda: [0,+ \infty) \times \R \to (0,+\infty)$ is measurable and satisfies
\begin{equation}\label{lambda_int}
	\mathbb{E}\left[\int_0^t\lambda(s, Y_s) \ud s\right]<\infty,
\end{equation}
for every $t \ge 0$.
Therefore, the process $\Lambda$ is non-decreasing and it satisfies $\Lambda_0=0$ and $\Lambda_t<\infty$ $\P$-a.s. for every $t \ge 0$.
Let $\eta=\{\eta_t,\ t \ge 0\}$ be a standard Poisson process (i.e. with intensity equal to $1$) independent of the Brownian motion $W^Y$, and consider the process  $N = \{N_t,\ t \ge 0\}$ defined as $N_t=\eta_{\Lambda_t}$, for every $t \ge 0$. Then, $N$ is a doubly stochastic Poisson process (see Lemma 4 and the discussion on stochastic random measures in \citep[Section 2.1]{grandell1991}). The process $\{\lambda(t,Y_{t}),\ t \ge 0\}$ is called the {\em intensity} of $N$ and  
condition \eqref{lambda_int} implies that $N$ is non-explosive. Moreover, the compensated process $\tilde{N}=\{ \tilde{N}_{t}, \ t \ge 0 \}$, given
by \begin{equation}
	\label{Nmg}
	\tilde{N}_t= N_t - \Lambda_t, \ \ t \ge 0,
	\end{equation}
	is a martingale (see \citep[Chapter II]{BREMAUD}).
The jump times of the process $N$ describe the claim arrival times and are represented through the increasing sequence of (nonnegative) random variables $\{T_n\}_{n \in \bN}$. Let $I\subset[0,+\infty)$ an arbitrary interval and let $\{Z_{n}\}_{n \in \bN}$ be a sequence of independent and $I$-valued random variables independent of $N$ such that for each $n \in \bN$, $Z_n$ is $\mathcal{F}_{T_n}$-measurable and $\mathbb{E}[e^{\eta Z_n}]<+\infty$ for every $n$ and for  $\eta>0$. 
For every $n \in \bN$, $Z_n$ indicates the claim amount at time $T_n$. The distribution of claim amounts is described by the map $F:[0,+\infty) \times \R \times I \longrightarrow [0,1]$ such that for each $(t,y) \in [0,+\infty) \times \R$, $F(\cdot,\cdot,z)$ is a distribution function, with $F(t,y,0)=0$.

The cumulative claim process $C=\{C_t,\ t \ge 0\}$ is given by
\begin{equation}
C_t=\sum_{n=1}^{N_t} Z_n, \quad t \ge 0,
\end{equation}
and represents the total losses of the insurance company due to claims. In the sequel, it will be useful to describe $C$ in terms of its associated random measure
defined as follows
\begin{align}
\label{rm}
m(\ud t,\ud z)& 
= \sum_{n \in \bN} \delta_{(T_{n},Z_{n})} (\ud t,\ud z),
\end{align}
where $\delta_{(t,z)}$ is the Dirac measure at point $(t,z)\in [0,+ \infty) \times [0,+ \infty)$; hence, the claim process $C$ reads as
\begin{equation}
C_{t} = \int_{0}^{t} \int_I z m(\ud s,\ud z),\quad t \ge 0.
\end{equation}
\begin{remark}\label{nu}
For reader's convenience, we recall a set of properties of the random counting measure $m(\ud t, \ud z)$. For every $A\subset I$, the process
$\{m((0,t]\times A)\}$ is a counting process that gives the number of claims with claim size in the set $A$. In particular, $m((0,t]\times I)=N_t$ is the total number of claims up to time $t$.

The dual predictable projection $\nu$ of the random measure $m(\ud t, \ud z)$ is given by
	\begin{equation}\label{nudef}
	\nu(\ud t, \ud z) =  F(t,Y_t,\ud z)\lambda(t,Y_t)\ud t.
	\end{equation}
Moreover,  for every non-negative, predictable random field $\Gamma=\{\Gamma(t,z) , \ t\ge 0, \ z \in I \}$,
such that	\begin{equation}
	\esp{\int_0^t \int_I \Gamma(s,z) \lambda(s,Y_{s})  F(s,Y_s,\ud z)  \ud s} < \infty,
	\end{equation}
for every $t\ge 0$, the process
	\begin{equation}
	\bigg{\{ } \int_0^t \int_I \Gamma(s,z) \Big( m(\ud s,\ud z) - F(s,Y_s,\ud z) \lambda(s,Y_s) \ud s\Big)  , \quad t \ge 0 \bigg{\}}
	\end{equation}
	is a martingale, see e.g. \citep[Chapter $VIII$, Theorem T$3$]{BREMAUD} for further details.
Consequently, it holds that
 \begin{equation}
	\esp{\int_0^t \int_I \Gamma(s,z) m(\ud s, \ud z)} = \esp{\int_0^t \int_I \Gamma(s,z) F(s,Y_s,\ud z) \lambda(s,Y_s) \ud s},
	\end{equation}
for every $t \ge 0$.
\end{remark}

From now on, we consider the following set of assumptions.
\begin{ass}\label{assZint}
It holds that 
\begin{equation}\begin{split}
		\esp{\int_0^t \int_I z \lambda(s,Y_s) F(s,Y_s, \ud z) \ud s   }< \infty, \qquad \esp{\int_0^t \int_I e^{\eta z} \lambda(s,Y_s)  F(s,Y_s, \ud z) \ud s  }< \infty, \\
		\esp{\int_0^t \int_I z e^{\eta z} \lambda(s,Y_s) F(s,Y_s, \ud z) \ud s  }< \infty, \quad \esp{\int_0^t \int_I z^2 e^{\eta z} \lambda(s,Y_s)  F(s,Y_s, \ud z) \ud s }< \infty,
	\end{split}
\end{equation}
for every $t \geq 0$ and $\eta>0$.
\end{ass}
Such integrability conditions will be used in some technical steps of the solution of the optimization problem. One of the consequences is that the cumulative claim process is non-explosive:
\begin{equation}\label{claimfiniti}
	\esp{C_t}=\esp{\int_0^{t} \int_I z m (\ud s, \ud z) } = \esp{\int_0^{t} \int_I z \lambda(s,Y_s) F(s,Y_s, \ud z) \ud s} < \infty,
\end{equation}
for every $t \ge 0$.
The insurance company receives premia and, in order to mitigate the risk exposure, reinsures part of its claims by continuously purchasing a proportional reinsurance contract.
We assume that both insurance and reinsurance premia depend on the index $Y$, and hence they are stochastic.
This is inline with the recent literature, see, e.g. \citep{delong2007mean,  cao2020optimal, ceci2022optimal}. It has also been observed in empirical studies, e.g. \citep{assa2022risk}, that having a dynamic premium would strongly reduce  the risk of insolvency.  Classical premium calculation principles can be extended to accommodate this assumption by {\em conditioning} on the value of $Y_t$. As remarked, for instance, in \citep{delong2007mean}, the conditional version of the expected value principle keeps the property that premium rate is proportional to the claim arrival intensity. This is also the case for the conditional variance principle.

We consider an insurance gross premium process of the form $\{a(t,Y_t),\ t \ge 0\}$, and a reinsurance contract of proportional type, with premium rate process $\{b(t,Y_{t},\Theta_{t}), \ t \ge 0 \}$, where $\Theta=\{\Theta_{t} ,\ t \ge 0\}$, represents the protection level, for some functions $a:[0,+ \infty) \times \R \to [0,+\infty)$ and $b:[0,+ \infty) \times \R \times [0,1] \to [0,+\infty)$.
In particular, at any time $t \ge 0$, $\Theta_t$ represents the percentage of losses which are covered by the reinsurance.
We assume that functions $a(t,y)$ and $b(t,y,\Theta)$ are  jointly continuous with respect to the pair $(t, y)$ and the triplet $(t,y,\Theta)$, respectively.
Throughout the paper, we will also assume the following integrability conditions
\begin{equation}\label{b_intcond}
\mathbb{E} \bigg[ \int_{0}^{t} a(s,Y_{s}) \ud s \bigg] < \infty, \qquad \mathbb{E} \bigg[ \int_{0}^{t} b(s,Y_{s},1) \ud s \bigg] < \infty,
\end{equation}
for every $t \ge 0$.
\begin{remark}\label{rem:premi}
Conditional versions of some classical premium calculation principle read as follows.  Under the expected value principle, for every $t \ge 0$ we get that
\[
a(t, Y_t)=(1+\delta^I)\lambda(t, Y_t)\int_I z F(t, Y_t, \ud z), \quad b(t, Y_t, \Theta_t)=(1+\delta^R) \Theta_t \lambda(t, Y_t)\int_I z  F(t, Y_t, \ud z),
\]
where $\delta^I>0$, $\delta^R >0$ represent the insurance and reinsurance safety loading respectively, and for the variance principle it holds that
\begin{align*}
a(t, Y_t)&=\lambda(t, Y_t)\left(\int_I z F(t, Y_t, \ud z)+\delta^I\int_I z^2  F(t, Y_t, \ud z)\right),\\
b(t, Y_t, \Theta_t)&=\Theta_t \lambda(t, Y_t)\left(\int_I z F(t, Y_t, \ud z)+\Theta_t\delta^R\int_I z^2 F(t, Y_t, \ud z)\right),
\end{align*}
for every $t \ge 0$. One of the drawbacks of these simple premium calculation principles is that the optimal reinsurance strategy does not explicitly depend on the claim arrival intensity, even in the conditional case. This does not happen for more sophisticated premium evaluation principles such as, e.g. the modified variance principle or the intensity-adjusted risk principle. In some of our numerical experiments we employ the modified  variance principle, under which premia are given by
\begin{align*} \label{mvp}
	a(t, Y_t)&=\lambda(t, Y_t)\int_I z  F(t, Y_t, \ud z)+\delta^I \frac{\int_I z^2  F(t, Y_t, \ud z)}{\int_I z  F(t, Y_t, \ud z)},\\
	b(t, Y_t, \Theta_t)&=\Theta_t \lambda(t, Y_t) \int_I z F(t, Y_t,\ud z)+\delta^R\Theta_t \frac{\int_I z^2 F(t, Y_t,\ud z)}{\int_I z F(t, Y_t,\ud z)},
\end{align*}
for every $t\ge 0$. 
This choice yields to $Y$-dependent Markovian optimal reinsurance strategies, adapted to the available information.
\end{remark}

We make the following set of assumptions that extends the usual natural hypotheses on premia to the stochastic case.

\begin{ass}\label{ipass}
The function $b(t,y,\Theta)$ has continuous partial derivatives $\pd{b(t,y,\Theta)}{\Theta}$, $\pds{b(t,y,\Theta)}{\Theta}$ in $\Theta \in [0,1]$ and it is such that \begin{itemize}
 \item[$(i)$] $b(t,y,0)=0$, for all $(t,y) \in [0,+ \infty) \times \R$, since the cedent does not need to pay for a null protection;
 \item[$(ii)$] $\pd{b(t,y,\Theta)}{\Theta}\ge0$, for all $(t,y,\Theta) \in [0,+ \infty) \times \R \times [0,1]$, because the premium is increasing with respect to the protection level;
 \item[$(iii)$] $b(t,y,1)>a(t,y)$, for all $(t,y) \in [0,+ \infty) \times \R$, for preventing a profit without risk;
\end{itemize}
In the sequel, $\pd{b(t,y,0)}{\Theta}$ and $\pd{b(t,y,1)}{\Theta}$ are understood as right and left derivatives, respectively.
\end{ass}

The insurance company surplus (or reserve) process $R^\Theta=\{R_t^\Theta,\ t \ge 0\}$ satisfies the SDE
\begin{equation}
 \label{surplus}
 \ud R_{t}^{\Theta}=a(t,Y_{t})\ud t - b(t,Y_{t},\Theta_{t}) \ud t - (1-\Theta_{t-}) \ud C_{t}, \quad R_{0}^{\Theta}=R_{0}>0.
\end{equation}
Conditions \eqref{b_intcond} imply in particular that the surplus process $R^\Theta$ is well defined and $\esp{R^\theta_t}<\infty$, for all $t \ge 0$.

The insurance company may invest its wealth in a financial market consisting of a locally risk-free asset with price process $S^0 = \{S^0_t,\ t \ge 0\}$ and a stock with price process $S=\{S_t, \ t \ge 0\}$. Here, we assume zero interest rate, that is, $S^0_t=1$ for every $t \ge 0$\footnote{The case of non-zero interest rate can be obtained using scaling arguments.}, and that $S$ satisfies the SDE
\begin{equation}\label{s}
\ud S_{t}=\mu(t,Y_{t})S_{t}\ud t + \sigma(t,Y_{t})S_{t} \ud W_{t}^S, \quad S_{0}=s>0,
\end{equation}
where the process $W^S=\{W_{t}^S,\ t \ge 0\}$ is a standard Brownian motion on $(\Omega,\F,\P;\bF)$, correlated with $W^Y$ with constant correlation coefficient $\rho\in [-1,1]$, and independent of the random measure $m(\ud t, \ud z)$.
The functions $\mu: [0,+\infty) \times \R \to \R$ and $\sigma: [0,+\infty) \times\R \to (0,+\infty)$, representing the drift and the volatility of the stock price process, respectively, are assumed to be measurable and such that the system of equations \eqref{y}-\eqref{s} admits a unique strong solution. Hence, the pair $(Y,S)$ is a Markov process. Sufficient conditions for existence and uniqueness of the solution to the system \eqref{y}-\eqref{s} can be found, e.g. in \citep[Theorem 5.2.1]{OKS}.
We also assume the Novikov condition
\begin{equation}\label{novikov}
	\esp {e^{\frac{1}{2}\int_0^t \left(\frac{\mu(s,Y_s)}{\sigma(s,Y_s)} \right)^2 \ud s}} < \infty,
\end{equation}
for every $t \ge 0$, which implies the existence of a risk-neutral measure for $S$ and ensures that the financial market does not admit arbitrage opportunities. 

Note that, our model specification allows to describe a possible dependence between financial and insurance markets, through two types of interaction. The first one is realized by assuming that both financial market coefficients and the claim arrival intensity are functions of a common stochastic factor, the process $Y$. Indeed, in the real-world, exogenous events of different nature (such as social, cultural, geographical  conditions, political decisions, natural events) may affect the average number of claims that an insurance company experiences, as well as the performance of portfolios negotiated in the market. The second kind of interaction can be viewed as an environmental contagion effect and it is due to the non-zero correlation between the Brownian motions  $W^S$ and $W^Y$ driving the stock price and the factor process dynamics, respectively.

\section{Investment and reinsurance under forward dynamic exponential utilities}\label{sec:results}

The insurance company, with an initial wealth $x_0$, subscribes a {\em proportional reinsurance} and invests continuously the remaining part of its wealth in the financial market, following a self-financing strategy. For every $t \ge 0$, we denote by $\Pi_{t}$ the total amount of wealth invested in the risky asset at time $t$, and hence $X_t - \Pi_t$ is the capital invested in the riskless asset at time $t$. We do not make any restriction on the process $\Pi$, i.e. $\Pi_t\in \R$ for every $t \ge 0$, meaning short-selling and  borrowing/lending from the bank account are allowed. 
For any reinsurance-investment strategy $H=(\Theta,\Pi)=\{(\Theta_t,\Pi_t),\ t\ge 0\}$, the wealth process $X^H=\{X_t^H,\ t \ge 0\}$  of the insurance company satisfies the SDE
\begin{align}
\ud X_{t}^H &= \ud R_{t}^\Theta + \Pi_{t} \frac{\ud S_{t}}{S_{t}} + (X_{t}^H-\Pi_{t})\frac{\ud S^0_{t}}{S^0_{t}} \\
& = \big{\{}  a(t,Y_{t}) - b(t,Y_{t},\Theta_{t}) +\Pi_{t}\mu(t,Y_{t}) \big{\}}  \ud t + \Pi_{t}\sigma(t,Y_{t}) \ud W_{t}^S - (1-\Theta_{t-})\ud C_{t},\label{wealth}
\end{align}
with $X_{0}^H=x_0 \ge 0$, being the initial wealth. The unique solution of the SDE \eqref{wealth} is given by
\begin{equation}\label{wealthsol}
\begin{split}	
X_t^H & = x_0 +\int_0^t  \left(a(s,Y_s) - b(s,Y_s,\Theta_s)+ \Pi_s\mu(s,Y_s) \right)\ud s + \int_0^t \Pi_s \sigma(s,Y_s) \ud W_s^S 	\\ & - \int_0^t \int_I  (1-\Theta_{s-}) z m(\ud s,\ud z),\quad t \ge 0.	
\end{split}
\end{equation}
%




We aim to study an optimal investment/reinsurance decision problem for the insurance company, following a forward approach. 
 Therefore,  we assume that the preferences of the insurance company are described by a dynamic utility function, which depends on time and possibly on some other additional stochastic drivers. 
The insurance company starts with today' specification of its initial utility, without pre-committing an investment horizon and a terminal utility function at time $0$ and then moves forward in time, modifying its preferences in relation to the available information, via a self-generating criterion, according to the following definition (see also Definition 2.1 in \citep{musiela2008optimal}).
\begin{definition}\label{def:FDU}
Let $s\ge 0$ be a normalization point. An adapted process $U(x, s)=\{U_t(x, s),\ t \ge s\}$ is a {\em dynamic performance process} normalized at $s$ if
\begin{itemize}
\item[a)] the function $x \to U_t(x,s)$ is increasing and concave for all $t \ge s$;
\item[b)] for every self-financing strategy $H$, and $T\ge t\ge s$ it holds that
\[
U_ {t}(X_t^H,s)\ge \esp{U_{T}(X_{T}^{H},s) | \F_{t}};
\]
\item[c)] there exists a self-financing strategy $H^*$ such that  for all $T\ge t\ge s$ it holds that
\[
U_ {t}(X_t^{H^*},s)= \esp{U_{T}(X_{T}^{H^*},s) | \F_{t}};
\]
\item[d)] at $t=s$,
\[
U_{s}(x, s)=u(x),
\]
where $u(x)$ is a concave and increasing function of wealth.
\end{itemize}
\end{definition}

In this paper, we assume that utility data (i.e. the function $u(x)$) is of exponential type with constant risk aversion $\gamma>0$, we set the normalization point $s=0$ for notational convenience, and define a {\em  forward dynamic exponential utility} as follows \footnote{Sometimes, forward dynamics utilities with $s=0$ are called {\em spot utilities}.}.
\begin{definition}\label{FDU}
An adapted stochastic process $\{U_{t}(x),\ t \ge 0\}$ is a {\em forward dynamic exponential utility (FDU)}, normalized at $0$, if for all $t, T$ such that $0 \le t \le T$, it satisfies the stochastic optimization criterion
\begin{equation}\label{FDUsm}
		U_ {t}(x) = \left \{\begin{array}{ll}
		-e^{-\gamma x}, & \qquad t=0,  \\
		\max_{H \in \A} \esp{U_{T}(X_{T}^{H}) | \F_{t}}, & \qquad t\ge 0,
		\end{array} \right.
\end{equation}
with $X^H$ given by \eqref{wealth} and $\mathcal A$ denotes the set of admissible strategies which is defined below.
\end{definition}

Notice that in equation \eqref{FDUsm} we have omitted the dependence on the normalization point which has been fixed to $s=0$, and we have simply written $U_t(x)$ in place of $U_t(x,0)$.

This definition reflects the fact that the insurance company tracks its risk preferences over time and its optimal strategy is associated with the martingale property along the optimal wealth trajectory.

By construction, we get that forward dynamic exponential utilities have two important features: (i) there is no constraint on the length of the
trading horizon; (ii) a forward dynamic exponential utility coincides with the dynamic value function it generates, at all intermediate times.

To represent the forward exponential utility, we penalize the classical exponential utility with a stochastic process that describes the insurance company dynamic preferences; this penalizing process depends on market coefficients, collected premia and paid premia but it may also be linked to other sources of risk which affect the financial-insurance market. Precisely,
the penalizing process $P=\{P_t, \ t \ge 0\}$ is defined as
\begin{equation}\label{P}
	P_t=\int_{0}^{t}g(s,X_s^H,Y_s) \ud s + \int_{0}^{t}h(s,X_s^H,Y_s) \ud W^P_s, \quad t \ge 0,
\end{equation} where $g:[0,+\infty) \times \R^2 \longrightarrow \R$ and $h:[0,+\infty) \times \R^2 \longrightarrow \R$ are two measurable functions such that
\begin{equation}
	\label{Pcoeff}
	\esp{\int_0^t |g(s,X^H_s,Y_s)| \ud s + \int_0^t h^2(s,X^H_s,Y_s) \ud s}< \infty,
\end{equation}	for every $t \ge 0$ and every $H=(\Theta,\Pi) \in [0,1] \times \R$. Here, $W^P=\{W^P_{t}, \ t\ge 0\}$ is a standard Brownian motion which is correlated with the other two Brownian motions $W^Y$ and $W^S$ introduced in Section \ref{sec:comb_model}. In particular, we denote by $\rho^Y \in [-1,1]$ the correlation coefficient between $W^P$ and $W^Y$ and by $\rho^S \in [-1,1]$ the correlation coefficient between $W^P$ and $W^S$.

Now, we define the set of admissible reinsurance-investment strategies.

\begin{definition}\label{def:admissible_strategies}
An {\em admissible strategy} is a pair of predictable processes $H=(\Theta, \Pi)= \{(\Theta_{t},\Pi_{t}), \ t \ge 0\}$, representing the proportion of reinsured claims and the total amount invested in the risky asset, respectively, such that $\Theta=\{\Theta_{t}, \ t \ge 0\}$ takes values in $[0,1]$ and $\Pi=\{ \Pi_{t}, \ t \ge 0\}$ is $\R$-valued and such that
\begin{equation} \label{int_ammiss}
\esp{\int_0^t\left(|\Pi_s| |\mu(s,Y_s)|+\Pi_s^2 \sigma^2(s,Y_s)\right) \ud s} < \infty,
\end{equation}
and $\esp{e^{-\gamma X^H_t-P_t}}<\infty$, for every $t \ge 0$. We denote by $\A$ the set of admissible strategies. Whenever controls are restricted to the time interval $[t,+ \infty)$, we will use the notation $\A_t$.
\end{definition}

Next, we describe the functions $g$ and $h$ given in \eqref{P}, that identify the penalizing process. Firstly, we assume that \begin{equation}\label{condconc}
	-\pds{b}{\Theta}(t,y,\Theta) < \gamma \lambda(t,y) \int_I e^{\gamma(1-\Theta)z}z^2 F(t,y, \ud  z),
\end{equation}
for every $(t,y,\Theta) \in [0,+\infty) \times \R \times [0,1]$, and let $\widehat{\Theta}$ be the unique solution of the equation
\begin{equation}\label{solr}
	\pd{b}{\Theta}(t,y,\Theta)  = \lambda(t,y) \int_I z e^{\gamma (1-\Theta)z}  F(t,y,\ud z),
\end{equation}
which exists in view of condition \eqref{condconc}. Set
\begin{equation}\label{optr}
	\bar{\Theta}(t,y) = \left \{\begin{array}{lll}
		0, \qquad \quad \ (t,y) \in \D_0 \\
		1, \qquad \quad \ (t,y) \in \D_1\\
		\widehat{\Theta}(t,y), \quad (t,y) \in (\D_0\cup \D_1)^c, \\
	\end{array} \right.
\end{equation}
where \begin{align}
	& \D_0 \equiv \bigg{\{} (t,y) \in [0,+\infty) \times \R \ \big{|} \ \lambda(t,y) \int_I ze^{\gamma z} F(t,y,\ud z) \le \pd{b}{\Theta}(t,y,0)   \bigg{\}}, \label{d0}\\
	& \D_1 \equiv \bigg{\{} (t,y) \in [0,+\infty) \times \R \ \big{|} \ \pd{b}{\Theta}(t,y,1) \le \lambda(t,y)\int_I z F(t,y,\ud z) \bigg{\}}\label{d1}
\end{align}
and $(\D_0\cup \D_1)^c$ is the complementary set of $\D_0\cup \D_1$.
We introduce the function $\varphi:[0,+ \infty) \times \R \to \R$ defined as
\begin{equation}\label{fi}
	\varphi(t,y) = \gamma b(t,y,\bar{\Theta}) + \lambda(t,y) \int_I \Big(e^{\gamma (1-\bar \Theta)z} -1 \Big) F(t,y,\ud z),
\end{equation}
for every $(t,y) \in [0,+ \infty)  \times \R$, with $\bar \Theta = \bar \Theta(t,y)$ given by \eqref{solr}.
Then, we assume that $g$ and $h$ satisfy
\begin{equation}\label{P_eq}
	g(t,x,y)=-\gamma a(t,y)+\frac{1}{2}h^2(t,x,y)-\frac{1}{2 \sigma^2(t,y)}\big(\mu(t,y)-\rho^S \sigma(t,y)h(t,x,y) \big)^2 + \varphi(t,y).
\end{equation}
\begin{remark}
	We observe that in view  of Assumption \ref{assZint} and \eqref{b_intcond}, $\mathbb{E}\left[\int_0^t \varphi(s, Y_s) \ud s\right]<\infty$, for each $t \ge 0$; therefore, if $\mathbb{E}\left[\int_0^t h^2(s, X_s^H, Y_s)\ud s\right]<\infty$, for each $t \ge 0$ and $H=(\Theta,\Pi) \in [0,1] \times \R$, then it implies that $\esp{\int_0^t |g(s, X_s^H, Y_s)|\ud s} < \infty$, for each $t \ge 0$, so that \eqref{Pcoeff} is satisfied. Indeed, for each constant control $H=(\Theta,\Pi) \in [0,1] \times \R$ and $t \ge 0$, we have
\begin{align}
&\mathbb{E}\left[\int_0^t |g(s, X_s^H, Y_s)|\ud s\right]\\
&\le \mathbb{E}\left[\int_0^t \Big\{\gamma |a(s,Y_s)|+\frac{1}{2}(1+(\rho^S)^2)h^2(t,X^H_s,Y_s)+\frac{\mu^2(s, Y_s)}{2 \sigma^2(s, Y_s)}\right. \\ &\qquad \left. + \rho^S \left|\frac{\mu(s, Y_s)}{\sigma(s,Y_s)}\right||h(s,X_s^H,Y_s)| +| \varphi(s, Y_s)|\Big\}\ud s\right]\label{ineq1}\\
&\label{ineq2}\le \mathbb{E}\left[\int_0^t \Big\{\gamma |a(s,Y_s)|+\Big(\frac{(1+(\rho^S)^2)}{2}+(\rho^S)^2\Big)h^2(t,X^H_s,Y_s)+\frac{3\mu^2(s, Y_s)}{2 \sigma^2(s, Y_s)} +| \varphi(s, Y_s)|\Big\}\ud s\right]\quad{}\\
&\label{ineq3}<\infty,
\end{align}
where we have used the triangular inequality in \eqref{ineq1}, the Cauchy-Schwarz inequality in \eqref{ineq2}, and the integrability conditions \eqref{b_intcond} and \eqref{novikov}.
\end{remark}

Our first goal is to prove that the process $\{U_{t}(x),\ t \ge 0\}$ defined as
\begin{equation} \label{goal}
U_t(x)=-e^{-\gamma x - P_t}, \quad (t,x) \in [0,+ \infty) \times \R,
\end{equation}
where $P$ is given in \eqref{P},
is a forward dynamic exponential utility, normalized at $0$.

Loosely speaking, the process $\left\{e^{-P_t},\ t \ge 0 \right\}$ can be interpreted as the density of a probability measure, which encompasses some of the characteristics of the combined market, as the Sharpe ratio $\ds \frac{\mu(t,Y_t)}{\sigma(t,Y_t)}$, insurance/reinsurance premia $(a(t,Y_t)-b(t, Y_t, \bar{\Theta}(t,Y_t) ))$ and the claims (both the arrival intensity and the sizes). It is important to notice that this measure also depends on the risk-aversion parameter $\gamma$ of the initial utility. The effect of the process $\left\{e^{-P_t},\ t \ge 0 \right\}$ is to penalize the standard utility in order to account for market features.

\begin{remark} \label{RDFU}
Since $P_t=\int_{0}^{t}g(s,X_s^H,Y_s) \ud s + \int_{0}^{t}h(s,X_s^H,Y_s) \ud W^P_s, \ t \ge 0$, different choices of the functions $g$ and $h$ will result in a different penalizing process.
For example, if $h(t,x,y)=0$, we are in the zero-volatility case and by \eqref{P_eq} the function $g$ does not depend on $x$ and is given by \begin{equation}
	g(t,y) = -\frac{1}{2} \bigg( \frac{\mu(t,y)}{\sigma(t,y)} \bigg)^2 - \gamma a(t,y)  + \varphi(t,y).
\end{equation}
We observe that this choice for $g$ also allows us to consider a function $h$ of form
\begin{equation}
	h(t,y)=-\frac{2 \rho^S}{1-(\rho^S)^2} \frac{\mu(t,y)}{\sigma(t,y)}.
\end{equation}
Other special cases are, e.g., $g(t,x,y)=\frac{1}{2}h^2(t,x,y)$ or $g(t,x,y)=\frac{1}{2}h^2(t,x,y)-\ds \frac{1}{2}\frac{\mu^2(t,y)}{\sigma^2(t,y)}-\gamma a(t,y) +\varphi(t,y)$, for every $(t,x,y) \in [0,+\infty) \times \R^2$.

The Brownian motion $W^P$ driving the dynamics of the penalizing process plays an important role. Two extreme cases are $W^P=W^S$ and $W^P=W^Y$ corresponding to the case where $P$ is perfectly (positively) correlated with the stock price $S$ and where $P$ is driven by the same risk source of $Y$.
In the most general setting, $W^P$ is correlated to both $W^S$ and $W^Y$. That is, the penalizing process includes the randomness coming from the stock price as well as the exogenous process that affects the price and the claims.
\end{remark}

The definition of the function $\varphi$, and hence of the function $g$, depends on the specific choice of $\bar\Theta$. Instead of taking $\bar\Theta$ as in \eqref{optr}, one could have taken $\bar{\Theta}(t,Y_t)=1$ for all $t \ge 0$: this choice corresponds to have $g(t,y)=-\ds\frac{1}{2}\bigg( \frac{\mu(t,y)}{\sigma(t,y)}\bigg)^2-\gamma(a(t,y)-b(t,y,1)$, which in turn leads to a dynamic utility that does not adjust for claims. Instead, setting $\bar{\Theta}(t,Y_t)=0$ for all $t \ge 0$ corresponds to set $g(t,y)= -\ds\frac{1}{2}\bigg( \frac{\mu(t,y)}{\sigma(t,y)} \bigg)^2 - \gamma (a(t,y) - b(t,y,0)) + \lambda(t,y) \int_I \Big(e^{\gamma z} -1 \Big) F(t,y,\ud z)$, which instead, implies that the penalizing process accounts for the whole claims amount. Our decision on the function $\bar{\Theta}(t,y)$ lies in the middle: in a certain sense, as we will see later, we would like to incorporate in the utility preferences the amount of claims that the insurance will not able to cover via the optimal strategy.

In the sequel we make the following integrability assumption.
\begin{ass}\label{h_nov}
For every $t \ge 0$, and every $H=(\Theta,\Pi) \in [0,1] \times \R$,
  \begin{equation}\label{novikov2}
	\esp {\ds e^{\frac{1}{2}\int_0^t h^2(s, X^H_s, Y_s) \ud s}} < \infty.
\end{equation}
\end{ass}

\begin{theorem}\label{familyfeu}
The process $\{U_t(x),\ t \ge 0 \}$, given for $x \in \R$ and $t \ge 0$ by
	\begin{equation} \label{FDUe}
	U_t(x)=-e^{-\gamma x - P_t},
	\end{equation}
	with the process $P$ defined in \eqref{P}, is a forward dynamic exponential utility, normalized at $0$.
\end{theorem}


\begin{proof}
To show the result we prove that the process $\{U_t(x),\ t \ge 0 \}$ defined in \eqref{FDUe} verifies Definition \ref{FDU} (or equivalently, Definition \ref{def:FDU} with the initial condition $u(x,p)=-e^{-\gamma x-p}$). By construction, for every $t \ge 0$ the random variable $U_t(x)$ is $\F_t$-measurable. Next, we need to show that for arbitrary $t, T$ such that $0 \le t \le T$, we have
\begin{equation}\label{eq:1}
-e^{-\gamma x - P_t} = \max_{H \in \A} \esp{-e^{-\gamma X_{T}^H - P_T} \Big{|} \F_t}.
\end{equation}
The equality \eqref{eq:1} implies that $\{U_t(x), \ t \ge 0\}$ is a supermartingale for all admissible strategies $H \in \mathcal A$, and a martingale along some strategy $H^*\in \mathcal A$.
In view of the Markov property of the process $(X^H, Y,P)$, 
we consider the function $u:[0,+\infty) \times \R^3 \to (-\infty,0)$ given by
\begin{equation}\label{vfass}
	u(t,x,y,p) := \max_{H \in \A_t} \mathbb{E}_{t,x,y,p} \Big[ -e^{-\gamma X_{T}^H -P_T} \Big],
\end{equation}
where $ \mathbb{E}_{t,x,y,p} $ denotes the conditional expectation given $X_{t}^H=x$, $Y_{t}=y$ and $P_t=p$, for every $(t,x,y,p) \in [0,T) \times \R^3$.

If $u(t,x,y,p)$ is $\mathcal{C}^1$ in $t$ and $\mathcal{C}^2$ in $(x,y,p)$, by applying  It\^o's formula, we get that the function $u(t,x,y,p)$ solves the problem
\begin{align}
\label{HJBu} &\max_{(\Theta,\Pi) \in [0,1] \times \R} \L^H u(t,x,y,p) = 0, \quad  (t,x,y,p) \in [0,\infty) \times \R^3,\\
\label{HJBufinal}& u(T,x,y,p)  =-e^{-\gamma x-p}, \quad  (x,y,p) \in \R^3,
\end{align}
with the operator $\L^H$ denoting the infinitesimal generator of the Markov process $(X^H,Y,P)$ associated with a constant control $H$, see \eqref{mgenxyp} in Appendix \ref{sec:tech_res}.

Now, we choose 
$H^*=(\Theta^*,\Pi^*)$ given by
\begin{equation}
\Pi^*_t=\ds \frac{\mu(t,Y_t)}{\gamma \sigma^2(t,Y_t)}-\rho^S\frac{h(t,X^{H^*}_t,Y_t)}{\gamma \sigma(t,Y_t)}
\end{equation}
and $\Theta^*_t=\bar{\Theta}(t, Y_t)$, for each $t \in [0,T]$, with $\bar\Theta(t,y)$ as in \eqref{optr}; equality \eqref{eq:1} holds and then the martingale property along $H^*$ is satisfied.

Next, we use a guess-and-verify approach and we show that the function
$u(t,x,y,p)$ given by
\begin{equation}\label{ansatz}
u(t,x,y,p)=u(x,p)=-e^{-\gamma x-p}, \quad (x,p) \in \R^2,
\end{equation}
provides the unique classical solution to
the problem \eqref{HJBu}-\eqref{HJBufinal}. Clearly, $u(x,p)=-e^{-\gamma x-p}$, with $(x,p) \in \R^2$, is $\mathcal{C}^\infty$ and it is easy to check that solves \eqref{HJBu}-\eqref{HJBufinal}. Moreover,
uniqueness follows from the Verification Theorem  (see Theorem \ref{thver} in Appendix \ref{app:verification}).
Indeed, condition (ii) of Theorem \ref{thver} is trivially satisfied since the function $u(x,p)$ does not depend on $y$. Then, we just need to show the following conditions: for every $t,T$ such that $0\le t \le T$,
\begin{equation*}
	\begin{split}
	&\mathbb{E} \bigg[ \int_{t}^{T\wedge \tau_n} \Big( \sigma(s,Y_{s})\Pi_{s} \pd{u}{x}(s,X_s^H,Y_s,P_s) \Big)^{2} \ud v \bigg] < \infty,	
	\\ &\mathbb{E} \bigg[ \int_{t}^{T\wedge \tau_n} \Big( h(s,X_s^H,Y_{s}) \pd{u}{p}(s,X_s^H,Y_s,P_s) \Big)^{2} \ud v \bigg] < \infty,	
	\\  &\esp{\int_I \! \int_{t}^{T\wedge \tau_n} \! \Big{|} u\big{(}s,X_{s-}^H-(1-\Theta_{s-})z,Y_s,P_{s-}) \big{)} - u(s,X_{s-}^H,Y_s,P_{s-}) \Big{|} \lambda(s,Y_{s})  F(s,Y_s,\ud z) \ud s } < \infty,
	\end{split}\end{equation*}
	with $I\subset[0,+\infty)$ being an arbitrary interval, for a suitable, non-decreasing sequence of stopping times $\{ \tau_n \}_{n \in \bN}$ such that $\lim_{n\to +\infty} \tau_n= +\infty$.
	Therefore, for every $n \in \mathbb N$, we define
\begin{equation}
	\tau_n= \inf \left\{s \in [t,T] : |P_s| > n \vee \ X_s^H<-n  \right\}.
	\end{equation}
Over the stochastic interval $\llbracket t, T\wedge \tau_n\rrbracket$, since $X^H$ and $P$ do not explode, there is $\bar n\in \mathbb N$, with $\bar{n}\le n$, such that $\gamma X_t^H+ P_t \ge -\bar{n}(\gamma +1)$.
Hence, we get that
	\begin{equation*}
	\begin{split}
	&\esp{ \int_t^{T\wedge \tau_n} \Big( \sigma(s,Y_{s})\Pi_{s} \pd{u}{x}(s,X_{s}^H,Y_s,P_s) \Big)^{2} \ud s } \\
	& \qquad = \esp{ \int_t^{T\wedge \tau_n}  \sigma^2(s,Y_{s})\Pi_s^2 \Big(\gamma e^{-\gamma X_{s}^H-P_s } \Big)^{2} \ud s }
	 \le \Big(\gamma e^{(\gamma +1) \bar n}\Big)^2 \esp{\int_t^{T} \sigma^2(s,Y_s)\Pi_{s} ^{2} \ud s } < \infty,
	\end{split}
	\end{equation*}
 since $\Pi$ is an admissible investment strategy. Moreover, we have that \begin{equation*}
 	\begin{split}
 		&\esp{ \int_t^{T\wedge \tau_n} \Big( h(s,X_s^H,Y_{s}) \pd{u}{p}(s,X_{s}^H,Y_s,P_s) \Big)^{2} \ud v } \\
 		& \qquad = \esp{ \int_t^{T\wedge \tau_n}  h^2(s,X_s^H,Y_s) \Big( e^{-\gamma X_{s}^H-P_s}  \Big)^{2} \ud s }
 		\le \Big(e^{(\gamma +1)\bar n}\Big)^2 \esp{\int_t^{T} h^2(s,X^H_s,Y_s)\ud s } < \infty,
 	\end{split}
 \end{equation*}
 thanks to \eqref{Pcoeff}. Finally, we have that
\begin{equation}
	\begin{split}
	& \esp{\int_{t}^{T\wedge \tau_n} \int_I \Big{|} u\big{(}s,X_{s-}^{H}-(1-\Theta_{s-})z,Y_s,P_{s-}) \big{)} - u(s,X_{s-}^H,Y_s,P_{s-}) \Big{|}  \lambda(s,Y_{s}) F(s,Y_s,\ud z) \ud s } \\
	& \qquad =  \esp{ \int_{t}^{T\wedge \tau_n} \int_I e^{-\gamma X_{s-}^H-P_{s-}}\Big{|} e^{\gamma (1-\Theta_{s-})z} -1\Big{|} \lambda(s,Y_{s})  F(s,Y_s,\ud z)\ud s}
	\\  & \qquad
	\le e^{(\gamma +1)\bar n} \esp{\int_{t}^{T} \int_I (e^{\gamma z}-1) \lambda(s,Y_{s})  F(s,Y_s,\ud  z) \ud s} < \infty,
	\end{split}
	\end{equation}
 by Assumption \ref{assZint}. 
 Therefore, Theorem \ref{thver} applies and the function $u(x,p)=-e^{-\gamma x-p}$ is the unique solution of the problem \eqref{HJBu}--\eqref{HJBufinal}.

We conclude that \eqref{eq:1} holds, and then $\{U_{t}(x)= -e^{-\gamma x -P_t},\ t \ge 0\}$ is a forward dynamic exponential utility.
\end{proof}
Now, we characterize the optimal reinsurance-investment policy for this family of forward dynamic exponential utilities.
\begin{proposition}\label{prop:optimal}
The optimal strategy $H^*=(\Theta^*,\Pi^*)$ is given by the optimal reinsurance protection level $\Theta^*=\{\Theta^*_t, \ t \ge 0\}$, where $\Theta^*_t=\bar \Theta_t=\bar \Theta(t,Y_t)$, with $\bar\Theta(t,y)$ defined by \eqref{optr}, and the optimal investment portfolio $\Pi^*=\{\Pi^*_t,\ t \ge 0\}$, where $\Pi^*_t=\Pi^*(t, X_t^{H^*},Y_t)$, with
\begin{equation}\label{opti1}
\Pi^{*}(t,x,y)= \frac{\mu(t,y)}{\gamma \sigma^2(t,y)} -\rho^S\frac{h(t,x,y)}{\gamma \sigma(t,y)},
\end{equation}
for every $(t,x,y) \in [0,+\infty) \times \R^2$.
\end{proposition}
The proof can be found in Appendix \ref{proof_optimal}.

\begin{remark}
We briefly discuss the economic interpretation of the regions $\D_0$ and $\D_1$ introduced in \eqref{d0} and \eqref{d1}, respectively, when we describe the optimal reinsurance strategy $\Theta_t^*=\bar\Theta_t=\bar \Theta(t,Y_t)$, with $\bar \Theta(t,y)$ defined by \eqref{optr}.
We could say that, if the reinsurance is not too much expensive (more specifically, if the price of an infinitesimal protection is below a certain dynamic threshold), then full reinsurance is optimal; if the reinsurance is too much expensive (that is, if the price of an infinitesimal protection is above a certain dynamic threshold), then null reinsurance is optimal; otherwise, the optimal reinsurance strategy is provided by \eqref{solr}, i.e. by equating the marginal reinsurance cost and the marginal gain.
	
\end{remark}
The optimal reinsurance level and the optimal investment portfolio do not depend on the normalization point, which is consistent with the theory of forward dynamic utilities (see e.g. \citep{MZ}). Moreover, the form of the value function and of the optimal investment strategy are explicit, in contrast to their backward counterpart (see also Section \ref{sec:comparison} below). The optimal investment strategy in the classical backward case is usually given in feedback form and hence depends on the value function, characterized in terms of the solution of a partial differential equation, which is parabolic, except for a few cases. In this perspective, the forward dynamic approach also has a computational advantage, as the optimal strategy presents only the myopic component and the value function has a closed-form expression.
Finally, we observe that a penalizing process with a different choice of $\bar\Theta$ would lead to the same optimal protection level. Our choice is therefore motivated by the fact that taking $\bar \Theta$ as in \eqref{optr} means that the forward utility accounts for the amount of claims that are not covered by the reinsurance, and hence represents a risk for the insurance company.

\section{Zero-volatility forward exponential utilities}\label{sec:zerovol}

In this section we assume that the penalizing process satisfies
$$
P(t)=\int_{0}^{t}g(s,Y_s) \ud s,
$$
for all $ t \ge 0$,\footnote{In this section we use the notation $P(t)$ to underline that the penalizing process is absolutely continuous with respect to the Lebesgue measure.} where  the function $g$ is given by
\begin{equation}\label{zerovolatility}
	g(t,y) = -\frac{1}{2} \bigg( \frac{\mu(t,y)}{\sigma(t,y)} \bigg)^2 - \gamma a(t,y)  + \varphi(t,y),
\end{equation}
for all $(t,y)\in [0,+\infty)\times \R$, with $\varphi(t,y)$ introduced in \eqref{fi}. This is a special case where the function $h(t,x,y)$ is set equal to zero in the penalizing process $P$ given in \eqref{P}, hence there is no additional stochastic part describing  utility preferences.

The following result characterizes the forward utility process and the corresponding optimal strategy in the zero-volatility case.

\begin{corollary}\label{familyfeu1}
The process $\{U_t(x),\ t \ge 0 \}$, given for $x \in \R$ and $t \ge 0$ by
	\begin{equation} \label{FDUexp}
	U_t(x)=-e^{-\gamma x - \int_0^tg(s,Y_s) \ud s},
	\end{equation}
	with $g(t, y)$ defined in \eqref{zerovolatility}, is a forward dynamic exponential utility, normalized at $0$. Moreover, the optimal strategy $H^*=(\Theta^*,\Pi^*)$ is given by the optimal reinsurance protection level $\Theta^*=\{\Theta^*_t, \ t \ge 0\}$, where $\Theta^*_t=\bar \Theta_t=\bar \Theta(t,Y_t)$, with $\bar \Theta(t,y)$ defined by \eqref{optr}, and the optimal investment portfolio $\Pi^*=\{\Pi^*_t,\ t \ge 0\}$, where $\Pi^*_t=\Pi^*(t, Y_t)$, with
\begin{equation}\label{opti}
\Pi^{*}(t,y)= \frac{\mu(t,y)}{\gamma \sigma^2(t,y)},
\end{equation}
for every $(t,y) \in [0,+\infty) \times \R$.
\end{corollary}

\begin{proof}
It follows from Theorem \ref{familyfeu} and Proposition \ref{prop:optimal} setting $h(t,x,y)=0$, for every $(t,x,y) \in [0,+\infty) \times \R^2$.	
\end{proof}

\subsection{Comparison with Backward Utilities Preferences}\label{sec:comparison}

Next, we compare the optimal strategy and the value of the optimal investment-reinsurance problem under dynamic forward zero-volatility utility to the classical backward utility, under the model setting outlined in Section \ref {sec:comb_model}. The first result provides the optimal strategy and the value function when we account for the backward exponential utility function $u^B(x)=-e^{-\gamma x}$.
Consider the backward reinsurance-investment problem
\begin{align}\label{pb:backward}
	\max_{H \in \mathcal A}\mathbb{E}\left[-e^{-\gamma X^H_T}\right],
\end{align}
where $T \in (0,+\infty)$ is a fixed time horizon which coincides with the end of the investment period. We introduce the Cauchy problem
\begin{equation}\label{eq:htilde}
	\left\{
	\begin{array}{ll}
		\pd{\phi}{t}(t,y;T)
		+ \pd{\phi}{y}(t,y;T) \left(\alpha(t,y)-\rho\ds \frac{\mu(t,y)}{\sigma(t,y)} \beta(t,y)\right)
		+ \ds \frac{1}{2} \frac{\partial^2\phi}{\partial y^2}(t,y;T) \beta^2(t,y), \\
		\qquad  + \ds \frac{1}{2}\bigg(\pd{\phi}{y}(t,y;T)\bigg)^2(1-\rho^2)\beta^2(t,y)-g(t,y)=0, &  (t,y) \in [0,T)\times \R,\\
		\phi(T,y;T)=0,  & \ y \in \mathbb{R}. 
	\end{array}
	\right.
\end{equation}

\begin{theorem}\label{thm:backward}
	Let $\phi(t,y;T)$ be the unique classical solution of the problem \eqref{eq:htilde}.
	Then, the value function corresponding to the problem \eqref{pb:backward} is given by $$V(t,x,y;T)=-e^{-\gamma x - \phi(t,y;T)}.$$
	Moreover, the optimal Markovian reinsurance-investment strategy $(\Theta^{*,B},\Pi^{*,B})=\{(\Theta_t^{*,B},\Pi_t^{*,B})=(\Theta^{*,B}(t,Y_t),\Pi^{*,B}(t,Y_t)),\ t \in [0,T]\}$ is given by
	\begin{align}\label{eq:strategia_backward}
	\Theta^{*,B}(t,y) & = \bar \Theta(t,y),\\
		\Pi^{*,B}(t,y) & =\frac{\mu(t,y)}{\gamma \sigma^2(t,y)} - \rho\frac{\beta(t,y) \ds\frac{\partial \phi(t,y;T)}{\partial y}}{\gamma \sigma(t,y)},
	\end{align}
	for every $(t,y) \in [0,T] \times \R$, where the function $\bar \Theta(t,y)$ is given in \eqref{optr}.
\end{theorem}
A sketch of the proof of Theorem \ref{thm:backward} is provided in Appendix \ref{app:thm_backward}.

We immediately notice that forward and backward preferences induce different investment strategies and equality holds only if the factor process $Y$ and the price process $S$ are driven by uncorrelated Brownian motions. In particular, the optimal investment strategy under forward performances consists only of the myopic component: this is a consequence of the fact that any changes in state of the market are absorbed by a utility function that updates forward in time according to the new conditions. The backward approach, instead, is based on the assumption that a future utility preference is set at the beginning of the investment period and does not change over time. Therefore, what needs to account for changes in market conditions must be the investment strategy which consists of a myopic part and an additional demand. The latter accounts for the part of risk correlated with the stock price, and disappears if stochastic movements of the factor process $Y$ and the stock price process are orthogonal\footnote{By orthogonality we mean that the martingales driving the processes $Y$ and $S$ have zero predictable quadratic covariation.}.
Considering the reinsurance strategies under the backward and the forward utilities, we see that they coincide. This can be explained at the mathematical level, using the same argument of the zero-correlation case. In fact the martingale driving the factor process $Y$, which affects the loss intensity, is orthogonal to the loss process. This is an interesting effect that arises due to the nature of the forward utility, even when the financial and insurance frameworks are not independent.

Concerning the value functions, in both cases the structure is affine in the wealth. However, the multiplicative components $e^{-\phi(t, y; T)}$ and $e^{-P(t)}$\footnote{Notice that the function $P$ also depends on $y$, which we omitted to keep the same notation as in the classical literature.}
gather the effect due to market changes in very different ways.
If we compare the implicit expression of $P(t)$ given by $\pd{P(t)}{t}=g(t,y)$ with the initial condition $P(0)=0$  and the PDE that characterizes $\phi(t,y;T)$ (see equation \eqref{eq:htilde}) we understand that both performance criteria account for market changes in an aggregate way. However, the backward value function estimates the future risk, whereas in the forward case the value function accounts for past observation. The functions $\phi(t,y;T)$ and $P(t)$ are evidently not easy to compare from the analytical point of view. As discussed in \citep{MZ} these two processes are related to well known martingale measures, namely the minimal martingale measure and the minimal entropy measure.


\subsection{Existence and uniqueness of a classical solution}
In the sequel, we provide sufficient conditions for existence and uniqueness of the solution to the PDE in \eqref{eq:htilde}, involved in the backward reinsurance-investment problem and, as a consequence, of the classical solution to the corresponding Hamilton-Jacobi-Bellman (in short HJB) equation.

We observe that, if Brownian motions $W^Y$ and $W^S$ have zero correlation, then the PDE in \eqref{eq:htilde} becomes linear and a solution exists under suitable conditions on model coefficients (see, e.g. \citep[Theorem 5.3]{pham1998optimal} or \citep[Theorem 1]{colaneri2019classical}).

We also notice that the PDE is linear even when there is a perfect correlation (either positive or negative) between the two Brownian motions that drive the stochastic factor $Y$ and the stock price $S$. Consequently, also in this case it is easy to obtain existence of a classic solution of the problem \eqref{eq:htilde}.
Otherwise (i.e. for $\rho \in (-1,1) \setminus \{0\})$, similar to \citep{ZAR_trasf}, we introduce a transformation given by \begin{equation}
	\phi(t,y)=\kappa \ln(\xi(t,y)), \quad (t,y) \in [0,T]\times \R,
\end{equation}
for a suitable parameter $\kappa \in \R\smallsetminus\{0\}$.
Differentiating yields \begin{equation}
	\pd{\phi}{t}(t,y)=\kappa \frac{1}{\xi(t,y)} \pd{\xi}{t}(t,y), \quad \pd{\phi}{y}(t,y)=\kappa\frac{1}{\xi(t,y)}\pd{\xi}{y}(t,y), \quad \pds{\phi}{y}=\kappa \frac{1}{\xi}\pds{\xi}{y}-\kappa \frac{1}{\xi^2(t,y)}\bigg(\pd{\xi}{y}(t,y)\bigg)^2.
\end{equation}
Substituting the above derivatives in \eqref{eq:htilde} gives the following PDE
\begin{equation}\label{eq:xi1}
	\begin{array}{ll}
		\pd{\xi}{t}(t,y)
		+ \pd{\xi}{y}(t,y) \left(\alpha(t,y)-\rho \ds \frac{\mu(t,y)}{\sigma(t,y)} \beta(t,y)\right)+ \ds \frac{1}{2} \frac{\partial^2\xi}{\partial y^2}(t,y) \beta^2(t,y) \\ \qquad \ds  + \frac{1}{2\xi(t,y)}\pd{\xi}{y}^2(t,y)\beta^2(t,y) \big(\kappa (1-\rho^2) -1 \big)  - g(t,y)\frac{\xi(t,y)}{\kappa}=0,
	\end{array}
\end{equation}
for all $(t,y) \in [0,T)\times \R$, with the final condition $\xi(T,y)=1$, for every $ y \in \mathbb{R}$.
The above expression suggests that if we take the parameter
\begin{equation}
	\kappa=\frac{1}{1-\rho^2},
\end{equation} then the PDE \eqref{eq:xi1} becomes linear. Indeed, the transformed solution $\xi(t,y)$ satisfies
\begin{equation}\label{eq:xi}
	\pd{\xi}{t}(t,y)
	+ \pd{\xi}{y}(t,y) \left(\alpha(t,y)-\rho\ds \frac{\mu(t,y)}{\sigma(t,y)} \beta(t,y)\right)+ \ds \frac{1}{2} \frac{\partial^2\xi}{\partial y^2}(t,y) \beta^2(t,y) - \ds g(t,y)\frac{\xi(t,y)}{\kappa}=0,
\end{equation} for all $(t,y) \in [0,T)\times \R$, with the final condition $\xi(T,y)=1$, for every $ y \in \mathbb{R}$.

Now, we provide some sufficient conditions that ensure the existence of a classical solution of the equation \eqref{eq:xi} with the associated final condition, applying Theorem $1$ of \citep{HEATHSCH}. They prove that there exists only one classical solution and also provide a probabilistic representation by means of the Feynman–Kac formula. Then, let us introduce a new probability measure $\widetilde \P$ equivalent to $\P$, by setting
\begin{equation}
	\frac{\ud \widetilde{\P}}{\ud \P} \Big{|}_{\F_t}=\widetilde{L}_t=e^{-\big( \frac{1}{2} \int_0^t \rho^2 \left(\frac{\mu(s,Y_s)}{\sigma(s,Y_s)} \right)^2 \ud s  +\int_{0}^t \rho \left| \frac{\mu(s,Y_s)}{\sigma(s,Y_s)}\right| \ud W_s^Y \big)}, \qquad t\in[0,T].
	\end{equation}
Condition \eqref{novikov} ensures that the process $\widetilde{L}=\{\widetilde{L}_t, \ t \in [0,T]\}$ is a $\P$-martingale.
By the Girsanov theorem, the process $\widetilde{W}^Y=\{\widetilde{W}_t^Y, \ t\in [0,T]\}$, defined as $\widetilde{W}^Y_t = W^Y_t + \rho \int_0^t \frac{\mu(s,Y_s)}{\sigma(s,Y_s)} \ud s$, for every $t \in[0,T]$, is a $\widetilde{\P}$-Brownian motion. Thus, the dynamics of the stochastic factor $Y$ under $\widetilde \P$ are given by
\begin{equation}
	\label{y_Ptilde}
	\ud Y_{t}=\left( \alpha(t,Y_{t}) -\rho \frac{\mu(t,Y_t)}{\sigma(t,Y_t)}\beta(t,Y_t) \right) \ud t + \beta(t,Y_{t}) \ud \widetilde{W}_{t}^{Y}, \quad  Y_{0}=y_0 \in \R.
\end{equation}

\begin{proposition}
	Suppose that functions $\mu(t,y)$, $\sigma(t,y)$, 
	$ \alpha(t,y)$, $\beta(t,y)$ are locally Lipschitz continuous in $(t,y) \in [0,T] \times \R$. Suppose also that $a(t,y)$, $b(t,y, \Theta)$ and $\lambda(t,y)$ are bounded and Lipschitz-continuous in $(t,y) \in [0,T] \times \R$.  	
	In addition, let us assume that $\beta(t,y) \geq \eta$, for all $(t,y) \in [0,T] \times \R$.
	Then, there exists a unique positive and bounded classical solution of equation \eqref{eq:xi} which is given by \begin{equation}\label{xi_FK}
		\xi(t,y)= \esp{e^{(1-\rho^2)\int_{t}^{T}g(s,Y_s)}\ud s},
	\end{equation} for all $(t,y) \in [0,T) \times \R$, with  terminal condition $\xi(T,y)=1$, for every $y \in \R$.
\end{proposition}
\begin{proof}
	We apply Theorem $1$ of \citep{HEATHSCH} to prove existence and uniqueness of classical solution to equation \eqref{eq:xi} by verifying that their required conditions $(A1)$, $(A2)$, $(A3a')$-($A3e')$ hold in our case.
	Consider a sequence of bounded sets $\{D_n=(-n,n), \ n \in \mathbb{N} \}$, such that $\ds \bigcup_{n\in \mathbb{N}} D_n= \R$. Since functions $\mu(t,y)$, $\sigma(t,y)$, 
	$ \alpha(t,y)$, $\beta(t,y)$ are locally Lipschitz continuous in $(t,y) \in [0,T] \times \R$, conditions ($A1$) and ($A2$) of \citep{HEATHSCH} for the coefficients $\alpha(t,y)-\rho \frac{\mu(t,y)}{\sigma(t,y)}\beta(t,y)$ and $\beta(t,y)$ are satisfied on $(t,y) \in [0,T] \times \R$. This also implies that $(A3a')$ holds. Moreover,  $\beta^2$ is uniformly elliptic on $(t,y) \in [0,T] \times\bar{D}_n$, for every $n \in \mathbb{N}$; i.e. ($A3b'$) holds. The H\"older-continuity of functions $a$, $b$ and $\lambda$ imply that also the function $g$, defined by \eqref{zerovolatility}, is H\"{o}lder-continuous, as required by ($A3c'$). Further, in our problem $g \equiv 0$ and $h\equiv 1$; thus, condition ($A3d'$) is trivially satisfied. To complete the proof, we check for ($A3e'$). Thanks to \citep[Lemma 2]{HEATHSCH}, it is sufficient to prove that the function $g$ is continuous and bounded from above, which is satisfied under our assumptions.
	Hence, the PDE \eqref{eq:xi} admits a unique classical solution $\xi(t,y)$ on $(t,y) \in [0,T] \times \R$ satisfying $\xi(T,y)=1$, given by \eqref{xi_FK}.
\end{proof}

\section{The conditional certainty equivalent}\label{sec:CCE}

In this section we discuss the concept of conditional certainty equivalent (in short CCE) introduced in  \citep{maggisfrittelli}, when we consider our family of forward exponential dynamic utilities. The definition of the CCE corresponds to the natural generalization to the dynamic and stochastic environment of the classical notion of the certainty equivalent, as given in \citep{pratt1976risk}.

\begin{definition}
Let $X^H=\{X_t^H, \ t \ge 0 \}$ be the wealth process corresponding to a constant strategy $H=(\Theta,\Pi) \in [0,1] \times \R$ and let $T>0$ be a finite time horizon. For every $0\leq t\leq T$, let $U_t(x)$ be the forward dynamic utility defined in \eqref{FDU}. Then, we define  the {\em conditional certainty equivalent} $C_t(X_T^H)$ as the random variable given by
\begin{equation}\label{CCE}
	C_t(X_T^H)  = U_t^{-1} \Big(\esp{U_T (X_T^H, 0) | \F_t },0\Big),
\end{equation} for every $t \in [0,T]$.
\end{definition}

The inverse exponential utility process $\{U_t^{-1}(x,0),\ t \ge 0\}$, and hence the CCE, is well defined (see \citep[Lemma 1.1 and Definition 1.1]{maggisfrittelli}) and satisfiest the properties below that can be directly derived from \cite[Proposition 1.1]{maggisfrittelli}: for every $0 \le t \le s \le T$ and every constant strategies $H, \widetilde H \in [0,1] \times \R$, we have
\begin{enumerate}[label=(\roman*)]
	\item $C_t(X_T^H) = C_t \big( C_s(X_T^H) \big)$.
	\item $C_t(X_t^H)=X_t^H$.
	\item If $C_s(X_T^H) \le C_s(X_T^{\widetilde H})$ then $C_t(X_T^H) \le C_t(X_T^{\widetilde H})$. \\ In particular, if $X_T^H \leq X_T^{\widetilde H}$ then $C_t(X_T^{H}) \le C_t(X_T^{\widetilde H})$ and if $X_T^H = X_T^{\widetilde H}$ then $C_t(X_T^{H}) = C_t(X_T^{\widetilde H})$.  
	\item If $U(x,0) = \{U_t(x,0), \ t \in [0,T]\}$ is decreasing in time, 
then $C_t(X_T^H) \le \esp{C_s(X_T^H)|\F_t}$  and $\esp{C_t(X_T^H)} \le \esp{C_s(X_T^H)}$. Moreover, $C_t(X_T^H) \le \esp{X_T^H|\F_t}$ and therefore $\esp{C_t(X_T^H)} \le \esp{X_T^H}$.
\end{enumerate}

Some important financial implications can be drawn from the properties listed above. The first important feature of the certainty equivalent, coming from property $(iii)$, is time consistency. That is, any two wealths with the same CCE at a given time $s$, have the same CCE at any time prior than $s$. Second, by property $(iv)$,  we get that for dynamic utilities that are decreasing in time, the CCE is increasing. Considering time-decreasing utilities reflects the impatience of the insurance company, namely the effective desire for accumulation and the perspective undervaluation of the future. In this case, the guaranteed amount that a company would accept not to take the risk, becomes larger and larger as time goes, as a consequence of a the fact that the company has a better perception of the combined market conditions, which is represented by a smaller utility. The inequality $C_t(X_T^H) \le \esp{X_T^H|\F_t}$ expresses the risk aversion of the company.

In the sequel, we focus on the case of the zero-volatility forward dynamic exponential utility, discussed in Section \ref{sec:zerovol}, which allows to better discuss some of the properties of the CCE.

First, in this case it is immediate to see the monotonicity property in time. Indeed,
if the function $g(t,y)$ defined in \eqref{zerovolatility} is non-negative for every $(t,y) \in [0,+\infty)\times \R$, then given a constant strategy $H=(\Theta,\Pi) \in [0,1] \times \R$, the forward exponential utility process \eqref{FDUexp} is decreasing in time, and hence
\begin{equation}
	C_t(X_T^H)= -\frac{1}{\gamma} \ln \Big( \esp{e^{-\gamma X_T^{H} - \int_0^T g(s, Y_s)\ud s} \Big{|} \F_t}  \Big) - \frac{1}{\gamma} \int_0^t g(s, Y_s)\ud s ,
\end{equation} for every $t \in [0,T]$, i.e. the CCE  is increasing in time.

Next, we will provide an additional comparison between forward dynamic and classical static backward exponential utilities (i.e. $U(x)=-e^{-\gamma x}$,  $x \in \R$), in terms of the CCE. We observe that CCE, for the static case is given by
\begin{equation}
	\widetilde C_t(X_T^H) = -\frac{1}{\gamma} \ln \Big( \esp{e^{-\gamma X_T^H } \Big{|}
\F_t}  \Big),
\end{equation} for every $t \in [0,T]$.

At time $t=0$, CCEs under forward exponential dynamic utility and backward exponential static utility  reduce, respectively, to
\begin{align*}
C_0(X_T^H)=-\frac{1}{\gamma} \ln \Big( \esp{e^{-\gamma X_T^H - \int_0^T g(s, Y_s)\ud s} }  \Big), \qquad
\widetilde C_0(X_T^H)=-\frac{1}{\gamma} \ln \Big( \esp{e^{-\gamma X_T^H }}  \Big).
\end{align*}
Then, we have the following implication.

\begin{lemma}\label{iff}
For every $t \in [0,T]$, the cumulative penalizing $P(T)-P(t) >0,$ $\P-$a.s. if and only if $C_t(X_T^H)>\widetilde C_t(X_T^H),$ $\P-$a.s., for any given constant strategy $H \in [0,1] \times \R$. In particular, $P(T)>0$ if and only if $C_0(X_T^H)>\widetilde C_0(X_T^H)$.
\end{lemma}

\begin{proof}
For every $t \in [0,T]$, we consider the difference
\begin{align*}
&C_t(X_T^H)-\widetilde C_t(X_T^H)= - \frac{1}{\gamma} \ln \Big( \esp{e^{-\gamma X_T^H -\int_t^Tg(s, Y_s) \ud s } \Big{|} \F_t}  \Big) + \frac{1}{\gamma} \ln \Big( \esp{e^{-\gamma X_T^H} \Big{|} \F_t}  \Big)
\end{align*}
This is larger than zero if and only if
\[\esp{e^{-\gamma X_T^H} \Big{|} \F_t}>\esp{e^{-\gamma X_T^H -\int_t^Tg(s, Y_s) \ud s } \Big{|} \F_t}.\]
Since $e^{-\gamma X_T^H}$ and $e^{-\int_t^Tg(s, Y_s) \ud s}$ are both nonnegative random variables then we get that the inequality holds if and only if $0<e^{-\int_t^Tg(s, Y_s) \ud s}<1$. Recalling that  $P(T)-P(t)=\int_t^Tg(s, Y_s)\ud s$ we get that $C_t(X_T^H)>\widetilde C_t(X_T^H)$ if and only if $P(T)-P(t)>0$.
Taking $t=0$ we obtain the second inequality for the CCEs at the initial time.
\end{proof}

Lemma \ref{iff} says that at $t=0$, the guaranteed amount that an insurance company endowed with a forward utility would accept not to take the risk of investing and reinsure its claims is larger than the guaranteed amount for an insurance company with static backward utility, i.e. $C_0(X_T^H)\ge \widetilde C_0(X^H_T)$ if and only if the total penalizing $P(T)=\int_0^T g(s, Y_s) \ud s$ is nonnegative. Indeed, in this case we will have that the forward utility is smaller than the backward utility, and hence that, at time $t=0$, the perception of the risk for the first company (the one with the forward utility) is smaller than for the second company (the one with the backward). Such condition is satisfied, e.g., when the penalizing process is increasing, that is when $g(t,y)>0$, for every $(t,y) \in [0,T] \times \R$.

Due to the structure of the penalizing process, it is not easy to compare, in general, the values of $C_t(X_T^H)$ and $\widetilde C_t(X_T^H)$. However, under specific conditions we can derive some considerations.

\begin{proposition}\label{propo:CCE_new}
If
\begin{equation}\label{eq:cond_i}
\frac{\mu^2(t,y)}{\sigma^2(t,y)} \ge 
- 2 \gamma a(t,y) + 2 \min\left\{b(t,y,1), \lambda(t,y) \int_{I}\Big(e^{\gamma z} - 1 \Big) F(t,y,\ud z) \right\},
\end{equation}
for every $(t,y) \in [0,T]\times\R$, then $\widetilde C_t(X_T^H) \ge  C_t(X_T^H)$, for each $t \in [0,T]$.

Otherwise, let $K:=\inf_{(t,y)\in [0,T] \times \R}\gamma b(t,y,\bar \Theta)+\lambda(t,y)\int_I\left(e^{\gamma(1-\overline \Theta)z}-1\right)  F(t,y,\ud z)>0$, with $\bar \Theta$ given by \eqref{solr}.  If  $a(t,y) < K/\gamma$ for all $(t,y) \in [0,T]\times \mathbb{R}$ and 
\begin{equation}\label{eq:cond_ii}
	\frac{\mu^2(t,y)}{\sigma^2(t,y)} \le 2 \left(-\gamma a(t,y) + K\right), 
	\end{equation}
for every $(t,y) \in [0,T]\times\R$, then $C_t(X_T^H) \ge  \widetilde C_t(X_T^H)$, for each $t \in [0,T]$.
\end{proposition}

\begin{proof}
Recall that, in the zero-volatility case, the function $g(t,y)$ is given by
\begin{align}
g(t,y)&=-\frac{1}{2}\frac{\mu^2(t,y)}{\sigma^2(t,y)}-\gamma a(t,y)+ f(t,y, \bar {\Theta}),
\end{align}
where we have set $f(t,y, \Theta)=\gamma b(t,y,\Theta)+\lambda(t,y)\int_I \left(e^{\gamma(1-\Theta)z}-1\right) F(t,y,\ud z)$, for each $\Theta \in [0,1]$. Notice that $f(t,y, \Theta)$ is convex in $\Theta$, then it admits minimum, which is given by $\bar{\Theta}$. Moreover, since $\Theta \in [0,1]$, by convexity we get that
\begin{equation}\label{ftheta}
	f(t,y, \bar\Theta)\le \min\left\{\gamma b(t,y,1), \lambda(t,y)\int_I\left(e^{\gamma z} - 1 \right)  F(t,y, \ud z)\right\},
\end{equation}
for every $(t,y) \in [0,T]\times\R$;
Then, by \eqref{eq:cond_i} and \eqref{ftheta}, we get that $g(t,y) \le 0$, for each $(t,y) \in [0,T]\times\R$, yielding $P(T)-P(t) \le 0$, for every $t \in [0,T]$.  Lemma \ref{iff} implies that $\widetilde C_t(X_T^H) \ge  C_t(X_T^H)$, for each $t \in [0,T]$.


On the other hand, we note that $f(t,y,\bar\Theta) >0$, for every $(t,y) \in [0,T] \times \R$\footnote{Indeed, if $\bar\Theta \neq \{0,1\}$, then $b(t,y,\bar \Theta)>0$ and $\int_I\left(e^{\gamma(1-\Theta)z}-1\right)  F(t,y,\ud z)>0$. For $\bar \Theta=1$ $f(t,y,1)=\gamma b(t,y,1)>0$ and for $\bar \Theta=0$ it holds that $f(t,y,0)=\lambda(t,y)\int_I \left(e^{\gamma z}-1\right)  F(t,y,\ud z)>0$.}. We let $K>0$ be the infimum over all $(t,y) \in [0,T] \times \R$ of the function $f(t,y,\bar\Theta)$. If $a(t,y) < K/\gamma$ for all $(t,y) \in [0,T]\times \mathbb{R}$ and \eqref{eq:cond_ii}, we get that 
\begin{align}
g(t,y) = -\frac{1}{2}\frac{\mu^2(t,y)}{\sigma^2(t,y)}-\gamma a(t,y)+ f(t,y, \bar {\Theta})
 >  0,
 \end{align}
which implies $P(T)-P(t) \ge 0$, for every $t \in [0,T]$. Finally, thanks to Lemma \ref{iff}, we have that $\widetilde C_t(X_T^H) \le  C_t(X_T^H)$, for each $t \in [0,T]$.
\end{proof}

Proposition \ref{propo:CCE_new} can be understood as follows: if the revenues from investment allow to cover for the payment of the claims, then the amount of guaranteed money an insurance company would accept today instead of taking a risk of getting more money at a future date, is larger when preferences are described by a forward dynamic utility.
In other words, if market conditions are favorable, the insurance company with backward static utility preferences has higher certainty equivalent, meaning that they are less willing to risk.
In the numerical section we present two examples where each of the two conditions of Proposition \ref{propo:CCE_new} is satisfied.

\section{Numerical analysis}\label{sec:numerics}

In this section we perform some numerical experiments in order to investigate some features of the optimal reinsurance-investment strategy and the optimal value process under forward exponential
utility preferences. We also discuss similarities and differences with the static backward case.

We assume that insurance and financial operations take place in one year, starting from today; this means that we analyze our theoretical results in a time interval $[0,T]$, with $T=1$.

The proposed model specification is rich enough to incorporate several stochastic factor models considered in the literature. In our first example, the stochastic factor process is chosen to follow a Vasicek model, i.e.
\begin{equation}
	\ud Y_t=(\alpha_1+\alpha_2Y_t)\ud t+\beta\ud W^Y_t, \quad Y_0=-0.2,
\end{equation}
with constant coefficients $\alpha_1=0.2$, $\alpha_2=-1$ and $\beta=0.1$.

We also suppose that the function $\lambda$  is given by  $\lambda(t,y)=\lambda_0 e^{y}$, for each $(t,y) \in [0,T] \times \R$, where $\lambda_0=k e^{-Y_0}$ with $k>0$,  which guarantees that the intensity of the claims arrival process $N$ is positive. For the sake of semplicity, all random variables $\{Z_n\}_{n \in \bN}$ have common distribution and the claim size distribution, specifically the claim size distribution is assumed as $\Gamma(\alpha_{{\tiny \Gamma}},\beta_{{\tiny \Gamma}})$, $\alpha_{{\tiny \Gamma}},\beta_{{\tiny \Gamma}}>0$. In particular, we set $k=1$, $\alpha_{{\tiny \Gamma}}=1$ and two different values of $\beta_{{\tiny \Gamma}}$, namely $\beta_{{\tiny \Gamma}}=2$ which corresponds to larger losses and $\beta_{{\tiny \Gamma}}=1/3$ to smaller claims.

We consider a risky asset price process with an affine appreciation rate and a uniformly elliptic Scott volatility, described by the following SDE
\begin{equation}\label{eq:stock_full}
	\ud S_t=\mu(Y_t)S_t\ud t+\sigma(Y_t)S_t\ud W^S_t, \quad S_0=1,
\end{equation} where \begin{equation}
\mu(y):=\mu_1+\mu_2y, \qquad \sigma(y)=:\bar{c}\sqrt{\epsilon_1 + e^{\epsilon_2 y}},
\end{equation} for every $y \in \R$.
The Brownian motions $W^S$ and $W^Y$ are correlated with correlation coefficient $\rho$.
We set the risk aversion coefficient to $\gamma=0.5$. \\
Firstly, we study the case where the reinsurance premium is calculated under the conditional modified variance principle, and given by
\begin{equation}\label{RP_NA}
	b(y,\Theta)=\alpha_{\Gamma}\beta_{\Gamma} \lambda(y) \Theta + \delta_{R} \beta_{\Gamma}\Theta,
\end{equation} for each $(y,\Theta) \in \R\times [0,1]$. As pointed in Remark \ref{rem:premi}, this premium calculation principle allows to keep the dependence on $Y$ in the optimal reinsurance strategy, which in turn, means that the strategy adapts to the index value.  
Regarding the insurance safety loading  $\delta_I$ and the reinsurance safety loading $\delta_R$, we have that $\delta_I<\delta_R<2\delta_I$; thus, we set $\delta_I=0.3$ and $\delta_R=0.5$.

In Figure \ref{THETA(Y)} we see that in case of small claims, null reinsurance might be optimal for almost all negative values of the stochastic factor, whereas in case of large claims a big percentage of claim is always reinsured.
\begin{figure}[ht]
	\includegraphics[width=8cm]{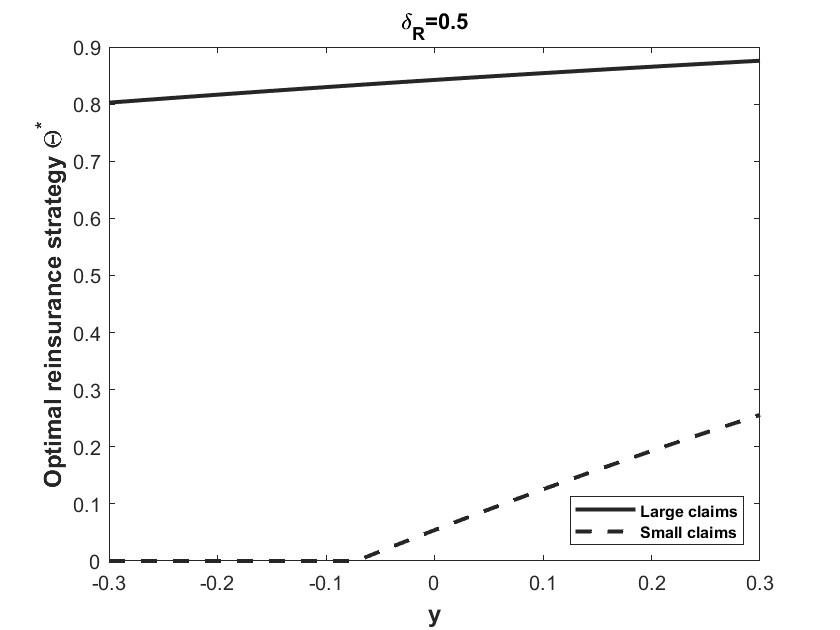}
	\caption{\label{THETA(Y)} Optimal proportional reinsurance strategy with respect to the values of the index, for large claims (solid line) and small claims (dashed line)}
\end{figure}

The protection level in case of big claims is larger than in the case when claims are smaller, that is, under the same claim arrival intensity, when expected claim amount is small, the insurance company purchases less reinsurance and for large losses it buys reinsurance with a higher protection level, to mitigate the risk.
Since $Y$ follows a Vasicek model with mean level $0.2$, starting from $-0.2$, this representation allows us to describe also the evolution of the protection level over time. Indeed, simulating the Brownian motion $W^Y$, we notice from Figure \ref{THETAtraj} that trajectories of the optimal reinsurance strategy are accumulated around $0.85$ in case of large claims and in the range $[0,0.25]$ in case of smaller claims.
\begin{figure}[ht]
	\includegraphics[width=8cm]{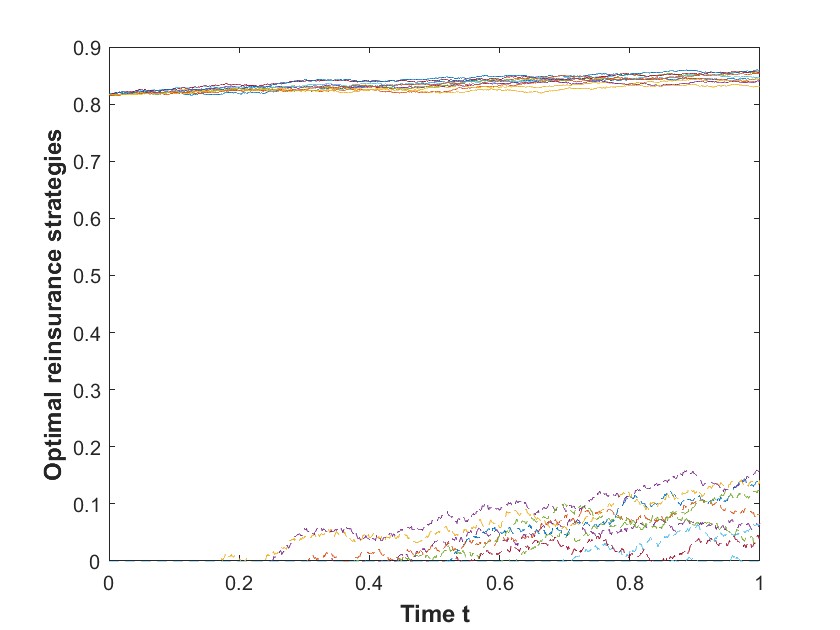}
	\caption{\label{THETAtraj} Some trajectory of the optimal proportional reinsurance strategy for large claims (solid line) and small claims (dashed line)}
\end{figure}

To better understand the form of the reinsurance strategy, in Figure \ref{THETA(Y)_EXP} and Figure \ref{THETA(Y)_cheap} we consider different safety loadings as the ones used in all the rest of this numerical section.  Using a reinsurance premium of the form \eqref{RP_NA}, we get that if reinsurance is expensive (which is the case, for example, of safety loadings at level $\delta_{R}=0.9$, $\delta_I=0.6$), the insurance company will reinsure fewer claims. In particular, as we see in Figure \ref{THETA(Y)_EXP}, in the case of small claims null reinsurance is optimal for every value of the index in the range $[-0.3,0.3]$ \footnote{The range of $y$ has been chosen according to the values in the simulations of the index $Y$.}.
\begin{figure}[ht]
	\includegraphics[width=8cm]{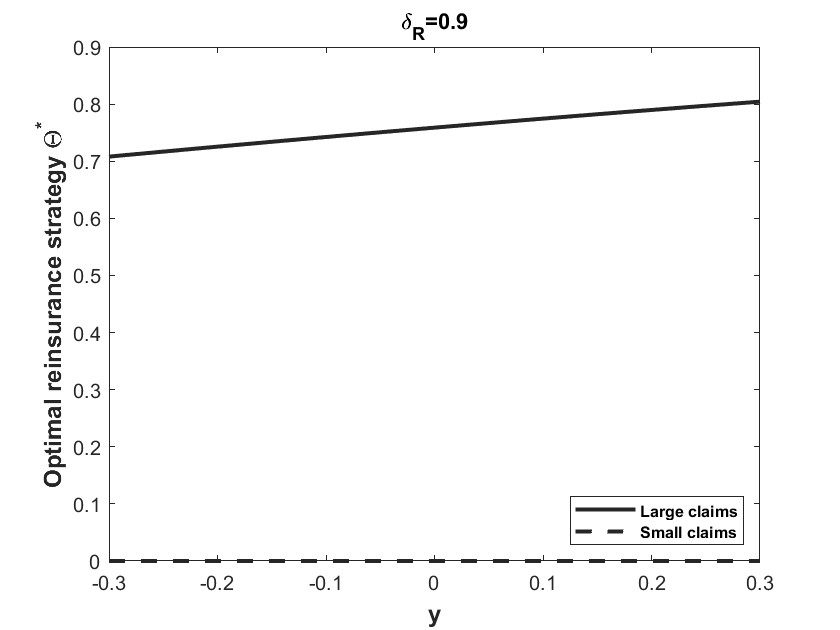}
	\caption{\label{THETA(Y)_EXP} Optimal proportional reinsurance strategy with respect to the values of the index, for large claims (solid line) and small claims (dashed line), with safety loadings $\delta_{R}=0.9$, $\delta_I=0.6$. }
\end{figure}

On the other hand, if reinsurance safety loading is small (for instance,  $\delta_R=0.1$, $\delta_I=0.07$), the insurance would opt for a larger protection level $\Theta$.
\begin{figure}[ht]
	\includegraphics[width=8cm]{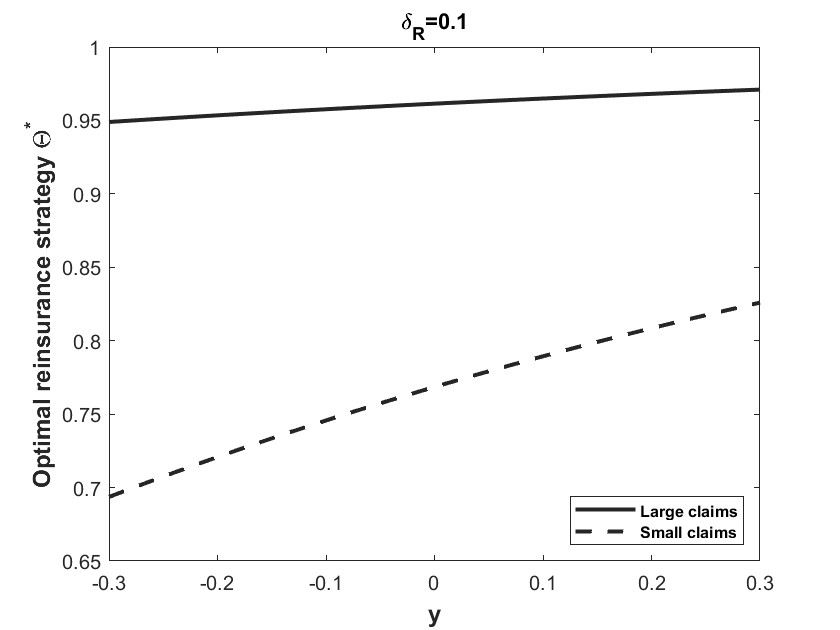}
	\caption{\label{THETA(Y)_cheap} Optimal proportional reinsurance strategy with respect to the values of the index, for large claims (solid line) and small claims (dashed line), with safety loadings $\delta_R=0.1$, $\delta_I=0.07$. }
\end{figure}

As can be seen from Figure \ref{THETA(Y)_cheap}, in both cases (large claims and small claims), under our parameter setting, full reinsurance and null reinsurance are never optimal. The reinsurance level, however, is always large, as reinsurance is very cheap.


Now, we restore the insurance safety loading at level $\delta_I=0.3$ and the reinsurance safety loading at level  $\delta_R=0.5$. We consider the optimal portfolio strategy. As pointed out by Remark \ref{RDFU}, different choices of the functions $g$ and $h$ lead to a different forward utility and, as a consequence, to a different optimal investment strategy. In the sequel, we consider the following choices of the function $h$, for every $(t,x,y)\in [0,T]\times \R^2$:
\begin{itemize}
	\item[(i)] $h_1(t,x,y)=0$. This corresponds to zero volatility utility function and the optimal investment strategy is the myopic strategy  $\Pi^*_1(y)=\frac{\mu(y)}{\gamma \sigma^2(y)}$;
	\item[(ii)] $\displaystyle h_2(t,x,y)=h_2(y)=-2 \frac{\rho^S}{1-(\rho^S)^2}\frac{\mu(y)}{\sigma(y)}$. In this case the insurance market only affects the drift of the penalizing process in the utility function but not its volatility. The optimal investment strategy is given by $\displaystyle \Pi^*_2(y)=\frac{\mu(y)}{\gamma \sigma^2(y)} + \frac{\mu(y)}{\gamma \sigma^2(y)} \frac{\rho^S}{1-(\rho^S)^2} $ and consists of a myopic component and an additional term that accounts for the correlation between the stock price and the preferences of the insurer;
	\item[(iii)] $\displaystyle h_3(t,x,y)=h_3(y)=\frac{\mu(y)}{\rho^S \sigma(y)}-\frac{1}{\rho^S} \sqrt{\varphi(y)-\gamma b(y,\Theta)}$. The optimal strategy for this choice is $\displaystyle \Pi^*_3(y)=\frac{1}{\gamma\sigma(y)}\sqrt{\varphi(y)-\gamma b(y,\Theta)}$. Interestingly, in this case the optimal strategy depends on uncovered claims and the market volatility but it is not affected by the Sharpe ratio and the insurance and reinsurance premia. The insurance company with this utility adjusts its preferences according to the state of the financial market and the premia that it receives and pays;
	\item[(iv)] $\displaystyle h_4(t,x,y)=-\bar{k}\frac{1}{\rho^S}\gamma \sigma(y) \bar{k}x+\frac{\mu(y)}{\rho^S \sigma(y)}$. Taking $h_4$ as an affine function of the wealth implies that the optimal strategy is of the form $\Pi^*_4(x)=\bar{k}x$, i.e. the insurance company would always invest in the risky asset the same percentage of the wealth. In particular, if $\bar{k}>1$ the insurance company would always borrow from the bank account and if $\bar{k}<0$ it always short-sells the risky asset.
\end{itemize}

In Figure \ref{ois}, we plot the optimal portfolio strategies corresponding to the first three choices of the function $h$, with respect to the value of the index. Taking, for example, $\beta_{{\tiny \Gamma}}=2$,  $\rho^S=0.5$ and $\bar{k}=0.5$, we consider two different parameter settings for the appreciation rate and the volatility of the stock. In the left panel we consider the case where the stock price volatility is highly affected by the index and the drift is almost constant. In this case the strategy $\Pi_3^*$ is less affected by the variation of the index and $\Pi_2^*$ is the most affected one.
In the middle panel, instead, the effect of the index on the drift dominates the effect on volatility. Here, we see that strategies are increasing for all $y$ in the range $[-0.3, 0.3]$.  $\Pi_2^*$ still assumes a wider range of values, hence it is the most affected by the variation of the index.
In the right panel we plot the strategies where both the drift and the volatility highly depend on the values of the index. Strategies $\Pi_1^*$ and $\Pi^*_2$ are not monotone, and $\Pi^*_2$ has the largest variation as for the other two cases.      
\begin{figure}[ht]
	\includegraphics[width=.3\textwidth]{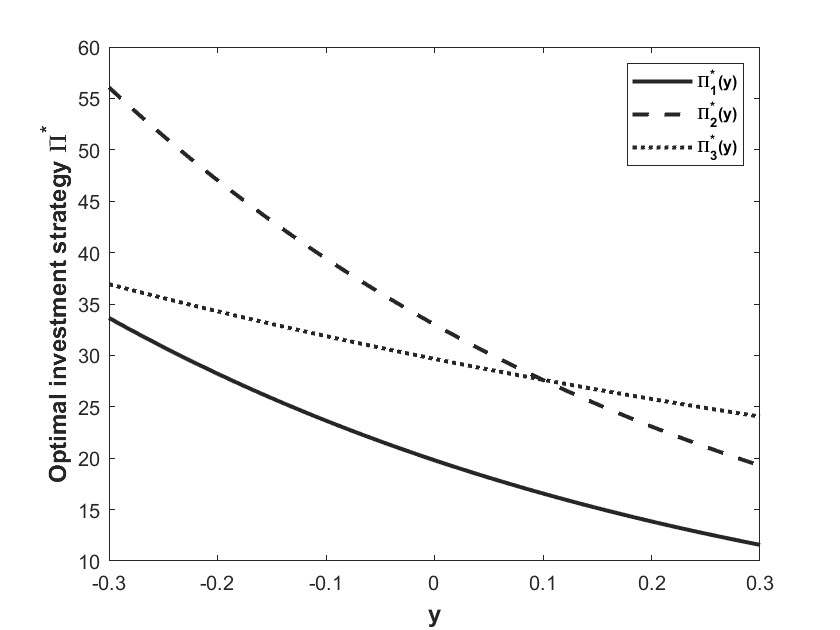}
	\includegraphics[width=.3\textwidth]{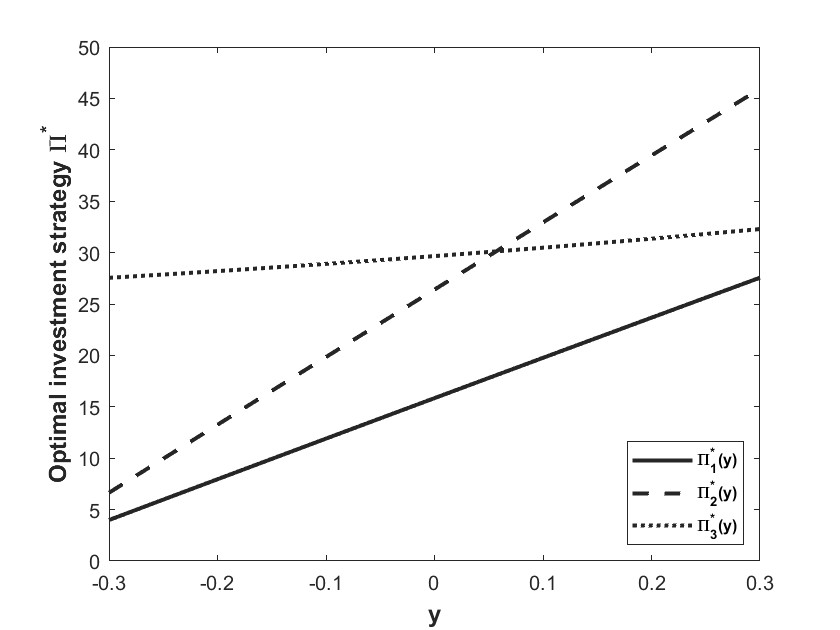}
	\includegraphics[width=.3\textwidth]{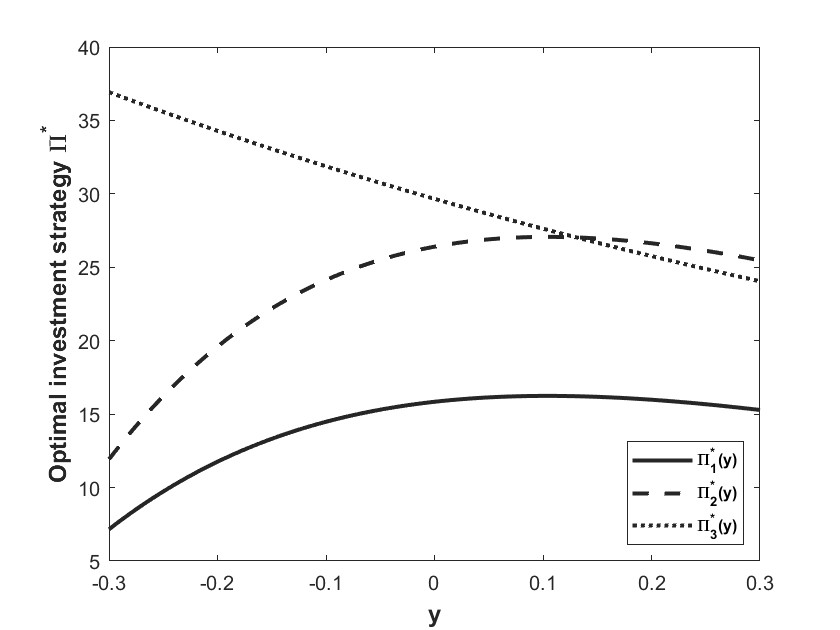}
	\caption{\label{ois} Optimal portfolio strategies with respect to the values of the index. under different forward preferences. Parameter settings: Left panel $\mu \in [0.094, 0.106]$, $\sigma\in [0.0748, 0.1354]$; Middle panel $\mu \in [0.02, 0.14]$, $\sigma\in [0.1002, 0.1008]$; Right panel $\mu \in [0.02, 0.14]$, $\sigma\in [0.0748, 0.1354]$}
\end{figure}


\subsection{CCE and comparison with backward investment}\label{example_comparison}
In this section we perform a comparison analysis between optimal strategies under forward and backward utilities. To simplify the framework we consider an example where the backward value function as well as the optimal strategy are characterized in closed form. Precisely, we employ the expected value principle to calculate the insurance and reinsurance premia, constant volatility for the stock price and  quadratic claim arrival intensity $\lambda(y)=\lambda_0(1+y+\frac{1}{2}y^2)$, so that it stays positive  \footnote{This choice of the function $\lambda$ corresponds to the second order approximation of $\lambda(y)=\lambda_0 e^{y}$, where $\lambda_0$ is a nonnegative constant, which is taken in the previous numerical experiments, and it is made so to obtain a unique classical solution of the HJB for the optimization problem under static backward utility.}.
Now, we suppose that claims follow a Gamma distribution $\Gamma(\alpha_{{\tiny \Gamma}},\beta_{{\tiny \Gamma}})$, where $\alpha_{{\tiny \Gamma}}=\beta_{{\tiny \Gamma}}=1$ (i.e. exponential distribution with mean $1$).

Moreover, we assume that the insurance and reinsurance premia are calculated under the expected value principle, that is
\begin{align}
	a(t,Y_t) &= (1+\delta_I) \esp{Z_1} \lambda(t, Y_t)\\
	b(t, Y_t, \Theta) &= (1+\delta_R) \lambda(t, Y_t) \esp{Z_1}\Theta_t.
\end{align}
For simplicity we denote $a=(1+\delta_I) \esp{Z_1}$ and $b^{(\Theta)}=(1+\delta_R) \esp{Z_1} \Theta{_t}$ and $\esp{Z_1}=\int_Iz  F(\ud z)$, where the insurance safety loading and the reinsurance safety loading now are set as $\delta_I=0.4$ and $\delta_R=0.7$, respectively.
We notice that in this case the optimal reinsurance strategy is given by $\Theta^{*,B}=\min\{1, \bar\Theta\}$ where $\bar \Theta$ is the unique solution of the equation $(1+\delta_R) \esp{Z_1}=\int_I ze^{\gamma z \Theta}  F(\ud z)$. Clearly, $\Theta^{*,B}$ does not depend on the stochastic factor $Y$.
We suppose that the dynamic of $S$ is given by
\begin{align}
	\ud S_t&=S_t (\mu_1+\mu_2Y_t)\ud t+ S_t \sigma \ud W^S_t, \quad S_0=1,
\end{align}
where $\sigma=\bar{c}\sqrt{\epsilon_1+1}$ (that corresponds to the Scott volatility above taking $\epsilon_2=0$). As for the other parameters, we take $\mu_1=0.08$, $\mu_2=0.2$, $\bar{c}=0.27$ and $\epsilon_1=0.01$.
We also denote $c(\Theta)=\int_I\big(e^{\gamma (1-\Theta )z}-1\big) F(\ud z)$.
In this case the function $g(t,y)$ is a quadratic function given by
\begin{align}
	&g(t,y)=-\frac{\mu_1^2}{2\sigma^2}-\big(\gamma a-\gamma b(\Theta^{*,B})-c(\Theta^{*,B})\big)\lambda_0(t) \\ &-\!\left(\frac{\mu_1\mu_2}{\sigma^2}
	\!+\!\big(\gamma a-\gamma b(\Theta^{*,B})-c(\Theta^{*,B})\big)\lambda_0(t)\right)y\!-\!\left(\frac{\mu_2^2}{2\sigma^2}\!+\!\big(\gamma a-\gamma b(\Theta^{*,B})-c(\Theta^{*,B})\big)\frac{\lambda_0(t)}{2}\right)y^2.
\end{align}
To solve the PDE \eqref{eq:htilde} we consider the following guess function: $$\phi(t,y;T)=\phi^{(0)}(t)+\phi^{(1)}(t)y+\phi^{(2)}(t) y^2.$$ Then, the optimal investment strategy becomes $\Pi^{*,B}(t,y)=\frac{\mu_1+\mu_2y}{\gamma \sigma^2}+\rho \frac{\beta(t) (\phi^{(1)}(t)+2\phi^{(2)}(t)y)}{\gamma \sigma}$.
Plugging the guess function in  \eqref{eq:htilde} and collecting the coefficients of $y$, $y^2$ and the constant term leads to the following system of ODEs
\begin{align}
	\pd{ \phi^{(0)}}{t}(t)=& - \phi^{(1)}(t) \big(\alpha_1(t)-\rho\frac{\mu_1\beta(t)}{\sigma}\big)-\frac{1}{2}\big(\phi^{(1)}(t)\big)^2(1-\rho^2)\beta^2(t)-\phi^{(2)}(t) \beta^2(t)\\
	&+\frac{\mu_1^2}{2\sigma^2}+\big(\gamma a-\gamma b(\Theta^{*,B})-c(\Theta^{*,B})\big)\lambda_0(t)\\
	\pd{ \phi^{(1)}}{t}(t)=&-\phi^{(1)}(t) \big(\alpha_2(t)-\rho\frac{\mu_2\beta(t)}{\sigma}\big)-2\phi^{(2)}(t)\big(\alpha_1(t)-\rho\frac{\mu_1\beta(t)}{\sigma^2}\big) \\&-2\phi^{(1)}(t)\phi^{(2)}(t) (1-\rho^2) \beta^2(t)+\frac{\mu_1\mu_2}{\sigma^2}
	+\big(\gamma a-\gamma b(\Theta^{*,B})-c(\Theta^{*,B})\big)\lambda_0(t)\\
	\pd{ \phi^{(2)}}{t}(t)=&-2\phi^{(2)}(t) \big(\alpha_2(t)-\rho\frac{\mu_2\beta(t)}{\sigma}\big)-2(\phi^{(2)})^2(1-\rho^2)\beta^2(t) \\
	&+\frac{\mu_2^2}{2\sigma^2}+\big(\gamma a-\gamma b(\Theta^{*,B})-c(\Theta^{*,B})\big)\frac{\lambda_0(t)}{2},
\end{align}
with the final conditions $\phi^{(0)}(T)=\phi^{(1)}(T)=\phi^{(2)}(T)=0$.
Regularity of the coefficients guarantees that a solution exists.

To highlight the impact of the stochastic factor on the optimal investment portfolio, we plot in Figure \ref{piy} the optimal investment strategies corresponding to forward and backward utilities as functions of the stochastic factor.

\begin{figure}
	\includegraphics[width=6cm,height=6cm]{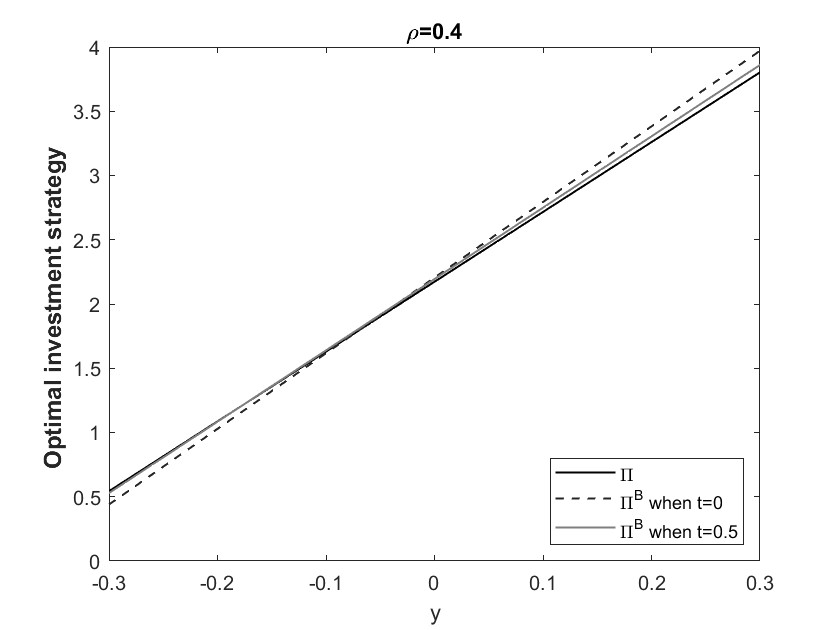}
	\includegraphics[width=6cm,height=6cm]{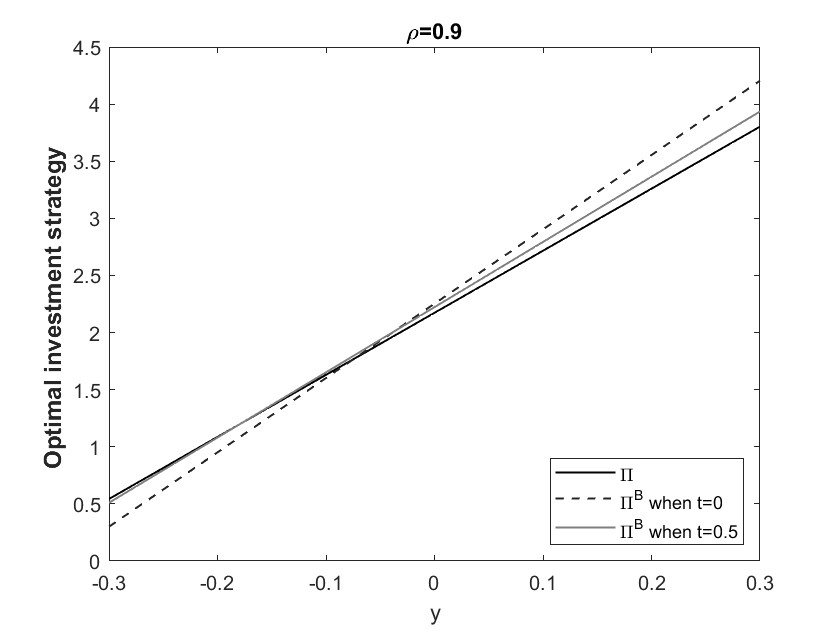}
	\caption{\label{piy} The optimal investment strategy when $\rho=0.4$ (left panel) $\rho=0.9$ (right panel), as functions of the stochastic factor. Solid (resp. dashed lines) line corresponds to the optimal portfolio under Forward (resp. Backward) utility. }
\end{figure}

We notice that the backward optimal portfolio is more sensitive to any variation of the stochastic factor with respect to the forward one, and this is amplified when the correlation coefficient is large. Moreover, we observe that this effect flattens out as maturity approaches. 
It is also clear that backward strategy gets closer and closer to the forward one, as the correlation coefficient approaches to zero, and they actually coincide when $\rho=0$. The difference between the optimal strategy under the forward utility and the optimal strategy under the backward utility decreases with time: indeed, as time to maturity reduces, also the estimates of future risk in the backward case has a smaller impact on the value function and hence on the optimal strategy that gets closer and closer to the myopic component.



We conclude with a brief analysis of the CCE. 
Figure \ref{RAt} provides a trajectory of the process $R_t:=\esp{X_T^H|\F_t}-C_t(X_T^H)$, for $t \in [0,T]$, which expresses the risk aversion of the company during the time interval. 
This process is decreasing and disappears at maturity.
\begin{figure}[ht]
	\includegraphics[width=8cm]{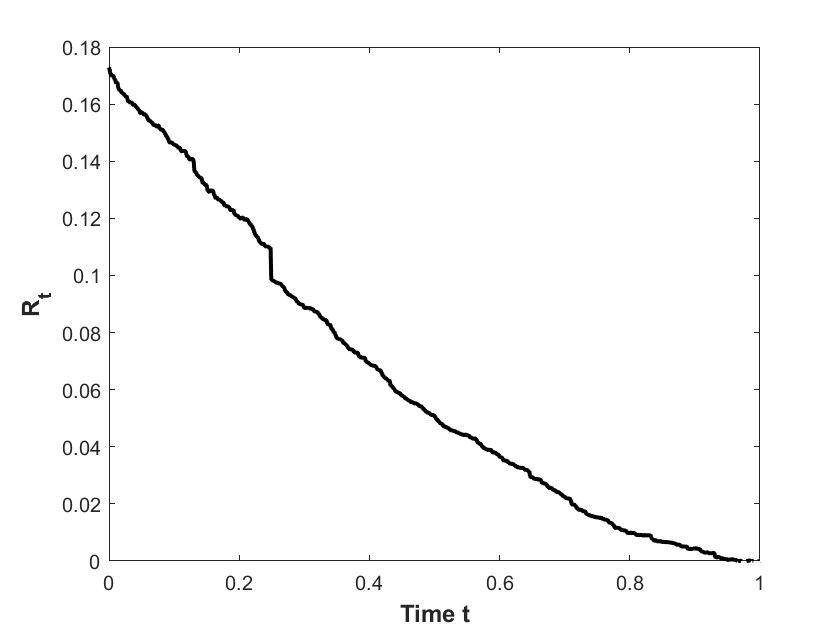}
	\caption{\label{RAt} The risk aversion of the company with respect to time.}
\end{figure}
Next we compare the CCE for the forward and the backward utility preferences.
To make the presentation more clear, we fix $\rho=0$.
Under the parameter setting that has been fixed in this section, it holds that $C_t(X_T^H)>\widetilde C_t(X_T^H),$ $\P-$a.s., due to the fact that $P(T)-P(t) >0,$ $\P-$a.s. (see Proposition \ref{iff}). This is confirmed in Figure \ref{CEF_CEB}. 
\begin{figure}[ht]
	\includegraphics[width=8cm]{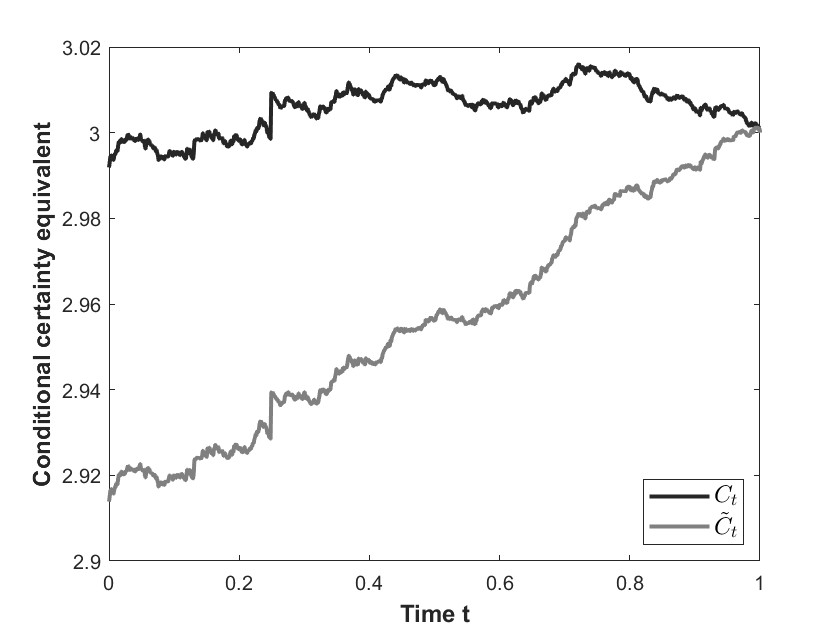}
	\caption{\label{CEF_CEB} CCEs under forward exponential dynamic utility and backward exponential static utility, with $\beta_{{\tiny \Gamma}}=1$ and $\bar{c}=0.27$.}
\end{figure}
For comparison purposes we also analyse the case of a smaller claim size and a smaller stock volatility, taking $\beta_{{\tiny \Gamma}}=1/3$ and $\bar{c}=0.1$. Figure \ref{CEB_CEF} shows that the relationship between the CCEs reverts and we get that  $C_t(X_T^H)<\widetilde C_t(X_T^H),$ $\P-$a.s..
\begin{figure}[ht]
	\includegraphics[width=8cm]{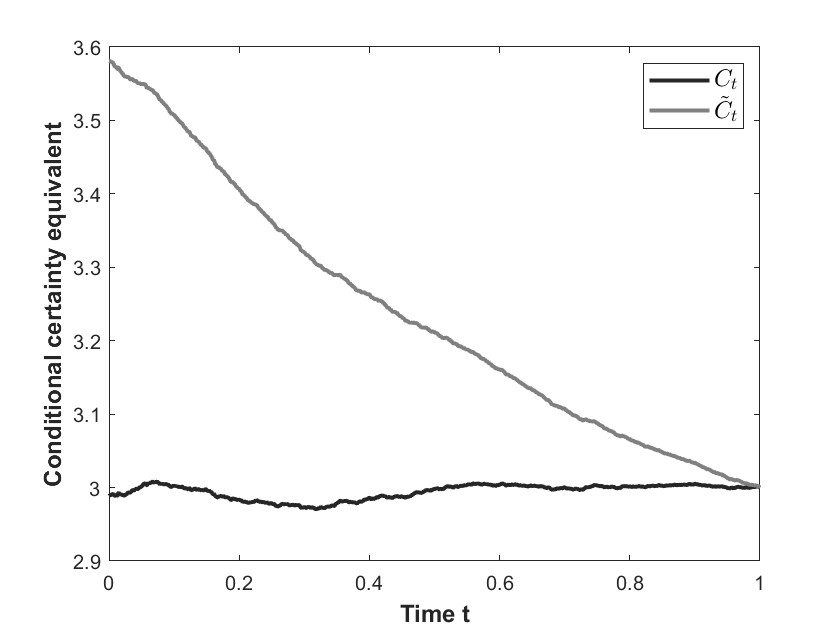}
	\caption{\label{CEB_CEF} CCEs under forward exponential dynamic utility and backward exponential static utility, with $\beta_{{\tiny \Gamma}}=1/3$ and $\bar{c}=0.1$.}
\end{figure}
In this case, the company does not buy reinsurance and instead invests a large part of its wealth in the risky asset. Consequently, a lower flexibility of the backward utility results into widely changes of the CCE in the backward case, which is also decreasing over the time interval.

\begin{center}
	{\bf Acknowledgments}
\end{center}
The authors are member of INdAM-Gnampa and their work has been partially supported through the Project U-UFMBAZ-2020-000791.

\appendix

\section{Assumptions and technical results}\label{sec:tech_res}

To solve the optimization problem \eqref{vfass}, we use the Markov property of the processes $(X^H,Y)$ and $(X^H,Y,P)$. 
Now, we compute, for every constant strategy $H=(\Theta,\Pi)\in [0,1] \times \R$, their infinitesimal generators that are useful in our paper.

We denote by $C^{1,2}$ the set of functions $f(t,x_1,...,x_n)$, which are differentiable with respect to time and twice differentiable with respect to variables $x_1,...,x_n$. 

We let $\tilde \L^H$ be the infinitesimal generator of the pair $(X^H,Y)$ which satisfies
	\begin{align}
			& \tilde \L^H f(t,x,y)\\
			 &\quad =\pd{f}{t}(t,x,y) + \big[ a(t,y) - b(t,y,\Theta) + \Pi\mu(t,y)\big]\pd{f}{x}(t,x,y) \\
			 &\quad + \frac{1}{2}\Pi^{2} \sigma^2(t,y) \pds{f}{x}(t,x,y) + \alpha(t,y)\pd{f}{y}(t,x,y) + \frac{1}{2}\beta^2(t,y)\pds{f}{y}(t,x,y) \\
			 & \quad + \rho \Pi \sigma(t,y) \beta(t,y)\pdsm{f}{x}{y}(t,x,y) \\
			 & \quad + \lambda(t,y) \int_I \Big{\{}f\big(t,x-(1-\Theta)z,y\big)-f(t,x,y)\Big{\}}  F(t,y,\ud z), \label{mgenxy}
		\end{align}
for every $(t,x,y) \in [0,+ \infty) \times \R^2 $ and for every function $f:[0,+ \infty) \times \R^2 \to \R$ in $C^{1,2}$ which is sufficiently integrable.  

Next, denote by $\L^H$ the infinitesimal generator of the triplet $(X^H,Y,P)$ which is given by
	\begin{align}
	& \L^H f(t,x,y,p)\\
	&\quad =\pd{f}{t}(t,x,y,p) + \big[ a(t,y) - b(t,y,\Theta) + \Pi\mu(t,y)\big]\pd{f}{x}(t,x,y,p) \\
	&\quad + \frac{1}{2}\Pi^{2} \sigma^2(t,y) \pds{f}{x}(t,x,y,p) + \alpha(t,y)\pd{f}{y}(t,x,y,p) + \frac{1}{2}\beta^2(t,y)\pds{f}{y}(t,x,y,p) \\
	& \quad + \rho \Pi \sigma(t,y) \beta(t,y)\pdsm{f}{x}{y}(t,x,y,p) + g(t,x,y)\pd{f}{p}(t,x,y,p) + \frac{1}{2}h^2(t,x,y)\pds{f}{p}(t,x,y,p) \\
	& \quad + \rho^S \Pi \sigma(t,y) h(t,x,y)\pdsm{f}{x}{p}(t,x,y,p) + \rho^Y \beta(t,y) h(t,x,y)\pdsm{f}{y}{p}(t,x,y,p)\\
	& \quad + \lambda(t,y) \int_I \Big{\{}f\big(t,x-(1-\Theta)z,y,p\big)-f(t,x,y,p)\Big{\}}  F(t,y,\ud z), \label{mgenxyp}
\end{align} for every $(t,x,y,p) \in [0,+ \infty) \times \R^3 $ and for every function $f:[0,+ \infty) \times \R^3 \to \R$ in $C^{1,2}$ which is sufficiently integrable. 

\subsection{The verification theorem}\label{app:verification}

Here, we prove a general verification result, which ensures that the function $u(t,x,y,p)$ defined in \eqref{ansatz} is the unique solution of the optimization problem \eqref{vfass}.

\begin{theorem}[Verification Theorem]\label{thver}
	Let
	$T\ge 0$ and let $\bar{u}: [0,T] \times \R^3 \to (-\infty,0)$ be a classical solution of the final value problem \eqref{HJBu}-\eqref{HJBufinal}\footnote{We recall that $\bar{u}(t,x,y,p)$ is a classical solution to \eqref{HJBu}-\eqref{HJBufinal} if the function $\bar{u}(t,x,y,p)$ solves \eqref{HJBu}-\eqref{HJBufinal} and it is continuous in $(t,x,y,p)$ and twice continuously differentiable with respect to $x,y,p$.}, which satisfies
	\begin{itemize}
		\item [\rm{(i)}] $\displaystyle{\esp {\int_{0}^{T} \Big( \sigma(r,Y_{r})\Pi_{r} \pd{\bar{u}}{x}(r,X_{r}^H,Y_r,P_r) \Big)^{2} \ud r }< \infty}$,
		\item [\rm{(ii)}] $\displaystyle{ \esp{ \int_{0}^{T} \Big( \beta(r,Y_r) \pd{\bar{u}}{y}(r,X_{r}^H,Y_r,P_r) \Big)^{2} \ud r }< \infty}$,
		\item [\rm{(iii)}] $\displaystyle{ \esp {\int_{0}^{T} \Big( h(r,X^H_r,Y_r) \pd{\bar{u}}{p}(r,X_{r}^H,Y_r,P_r) \Big)^{2} \ud r} < \infty}$,
		\item [\rm{(iv)}] $\displaystyle{  \esp{\int_{0}^{T} \! \lambda(r,Y_{r}) \int_I \Big{|} \bar{u}\big{(}r,X_{r-}^H-(1-\Theta_{r-})z,Y_r,P_{r-} \big{)} - \bar{u}(r,X_{r-}^H,Y_r,P_{r-}) \Big{|}  F(r,Y_{r},\ud z) \ud r} < \infty}$.
	\end{itemize}
	\begin{itemize}
		\item [$(a)$] Hence, $u(t,x,y,p) \le \bar{u}(t,x,y,p)$, for every admissible control $H \in \A$ and for every $(t,x,y,p) \in [0,T] \times \R^3$.
		\item [$(b)$] Moreover, if $\bar{u}(T,x,y,p)=u(T,x,y,p)$, for every $(x,y,p) \in \R^3$ and there exists $H^* \in \A$ such that $\L^{H^*}\bar{u}(t,x,y,p)=0$, for every $(t,x,y,p) \in [0,T) \times \R^3$, then $u=\bar{u}$ in $[0,T] \times \R^3$.
	\end{itemize}
\end{theorem}

\begin{proof}
	Let $H \in \A$ be an admissible control. Using equations \eqref{wealth} and \eqref{y} and applying It\^{o}'s formula to $\bar{u}(t,X_t^H,Y_t,P_t)$, we have that
	\begin{align}
		& \bar{u}(T,X_{T}^H,Y_{T},P_T) = \bar{u}(t,x,y,p) + \int_{t}^{T} \L^H \bar{u}(r,X_r^H,Y_r,P_r) \ud r \\ & \qquad + \int_{t}^{T} \Pi_{r}\sigma(r,Y_{r})\pd{\bar{u}}{x}(r,X_{r}^H,Y_{r},P_r)  \ud W_{r}^S + \int_{t}^{T} \beta(r,Y_{r})\pd{\bar{u}}{y}(r,X_{r}^H,Y_{r},P_r)  \ud W_{r}^Y \\
		&  \qquad + \int_{t}^{T}h(r,X_r^H,Y_r) \pd{\bar{u}}{p}(r,X_{r}^H,Y_{r},P_{r})  \ud W_{r}^P \\
		&  \qquad+\int_{t}^{T}\int_{I} \left(\bar{u} \big(r,X_{r-}^H -(1-\Theta_{r-})z,Y_r,P_{r-} \big) - \bar{u}(r,X_{r-}^H,Y_r,P_{r-}) \right)\\
		&  \qquad \quad  \times (m(\ud r, \ud z)- \lambda(r,Y_r)  F(r,Y_r,\ud z) \ud r),\label{itoV}
	\end{align}
	where $\mathcal{L}^H$ is introduced in \eqref{mgenxyp}. Let $M=\{M_t,\ t \in [0,T]\}$ be the stochastic process given by
	\begin{align}
		M_t & = \int_{0}^{t}\Pi_{r}\sigma(r,Y_{r})\pd{\bar u}{x}(r,X_{r}^H,Y_{r},P_r)  \ud W_{r}^S + \int_{0}^{t}\beta(r,Y_{r})\pd{\bar u}{y}(r,X_{r}^H,Y_{r},P_r)  \ud W_{r}^{Y} \\ & + \int_{t}^{T}h(r,X_r^H,Y_r) \pd{\bar{u}}{p}(r,X_{r}^H,Y_{r},P_{r})  \ud W_{r}^P \\ &  + \int_{0}^{t}\!\int_{I} \!\! \left(\bar{u} \big(r,X_{r-}^H\! -(1-\Theta_{r-})z,Y_r,P_{r-} \big) - \bar{u}(r,X_{r-}^H,Y_r,P_{r-}) \right) (m(\ud r, \ud z)- \lambda(r,Y_{r})  F(r,Y_r,\ud z)\ud r) 	\label{MG}
	\end{align} and observe that integrability conditions (i), (ii), (iii), (iv) ensure that the process $M$ is a martingale.
	Now, since $\bar{u}$ solves equation  \eqref{HJBu} with final condition \eqref{HJBufinal}, we get
	\begin{align}\label{itoV2}
		& \bar{u}(T,X_{T}^H,Y_{T},P_T) \leq \bar{u}(t,x,y,p)  + \int_{t}^{T} \Pi_{r}\sigma(r,Y_{r})\pd{\bar{u}}{x}(r,X_{r}^H,Y_{r},P_r)  \ud W_{r}^S \\ & \quad + \int_{t}^{T} \beta(r,Y_{r})\pd{\bar{u}}{y}(r,X_{r}^H,Y_{r},P_r) \ud W_{r}^Y + \int_{t}^{T}h(r,X_r^H,Y_r) \pd{\bar{u}}{p}(r,X_{r}^H,Y_{r},P_{r})  \ud W_{r}^P  \\
		&  \quad +\int_{t}^{T}\int_{I} \left(\bar{u} \big(r,X_{r-}^H -(1-\Theta_{r-})z,Y_r,P_{r-} \big) - \bar{u}(r,X_{r-}^H,Y_r,P_{r-}) \right)(m(\ud r, \ud z)- \lambda(r,Y_r)  F(r,Y_r,\ud z) \ud r),
	\end{align}
	for every $H \in \A$.
	Thus, taking the conditional expectation with respect to $X_t^H=x$, $Y_t=y$ and $P_t=p$ on both sides of inequality \eqref{itoV2}, leads to \begin{equation}
		\mathbb{E}_{t,x,y,p} \Big[ \bar{u}(T,X_{T}^H,Y_{T},P_T) \Big] \le \bar{u}(t,x,y,p).
	\end{equation} By the final condition in equation \eqref{HJBufinal}, we obtain \begin{equation}
		\mathbb{E}_{t,x,y,p} \Big[ -e^{-\gamma X_T^H - (P_T-P_t)} \Big] \le \bar{u}(t,x,y,p),
	\end{equation} for every $H \in \A$. Hence, $u(t,x,y,p) \le \bar{u}(t,x,y,p)$, as we wanted. Finally, we observe that if $H \in \A$ is the maximizer in equation \eqref{HJBu} with final condition \eqref{HJBufinal}, then the inequality above becomes an equality, and we obtain statement $(b)$, which concludes the proof.
\end{proof}

\subsection{Proof of Proposition \ref{prop:optimal}} \label{proof_optimal}

We consider the optimization problem defined by \eqref{HJBu}--\eqref{HJBufinal}. Using the form of the function $u$, we observe that the problem can be written as
\begin{align}\label{HJBp}
&-\gamma a(t,y)u(x,p) -  g(t,x,y)u(x,p)  + \frac{1}{2}h^2(t,x,y)u(x,p) \\ &+\max_{\Theta \in [0,1]}\Psi_1(\Theta,t,x,y,p)  + \max_{\Pi \in \R} \Psi_2(\Pi,t,x,y,p)  =0,
\end{align}
for all $(t,x,y,p) \in [0,+\infty) \times \R^3$ with the final condition $u(T,x,y,p)=-e^{-\gamma x-p},$ for all $(x,y,p) \in \R^3$, where the functions $\Psi_1$, $\Psi_2$ are defined as
\begin{align}
\Psi_1(\Theta,t,x,y,p) & = \gamma b(t,y,\Theta) u(x,p) + \lambda(t,y) \int_I u(x,p)(e^{\gamma (1-\Theta)z}-1)  F(t,y,\ud z),\label{psid}\\
\Psi_2(\Pi,t,x,y,p) & = -\gamma \Pi \mu(t,y) u(x,p) + \frac{1}{2} \gamma^2 \Pi^{2}\sigma^2(t,y) u(x,p) + \gamma \rho^S \Pi \sigma(t,y) h(t,x,y) u(x,p) . \label{psipi}
\end{align}
Now, we compute the optimal protection level $\Theta^*$. The function $\Psi_1$ is continuous in $\Theta$, due to the assumptions on the function $b(t,y,\Theta)$, and $\Theta\in [0,1]$, therefore a maximum exists. The first and second order derivatives of $\Psi_1$ are respectively given by
\begin{equation}
\pd{\Psi_1}{\Theta}(\Theta,t,x,y,p)= -\gamma u(x,p) \left\{  \pd{b}{\Theta}(t,y,\Theta) - \lambda(t,y) \int_I e^{\gamma (1-\Theta)z}z   F(t,y,\ud z) \right\},
\end{equation}
\begin{equation}
\pds{\Psi_1}{\Theta}(\Theta,t,x,y,p)= - \gamma  u(x,p) \left\{ \pds{b}{\Theta}(t,y,\Theta) + \gamma  \lambda(t,y) \int_I e^{\gamma (1-\Theta)z}z^2  F(t,y,\ud z)\right\},
\end{equation}
and they are continuous in $\Theta$ and well defined thanks to Assumptions  \ref{assZint} and \ref{ipass}. 
In virtue of condition \eqref{condconc}, $\Psi_1(\Theta,t,x,y,p)$ is strictly concave in $\Theta$ and hence it admits a unique maximizer $\Theta^{*} \in [0,1]$, whose measurability follows by classical selection theorems.
Let  $\widehat \Theta$ be the solution of the equation
$\pd{\Psi_1}{\Theta}(\Theta, t,x,y,p)=0$. If $\widehat \Theta\in (0,1)$, then $\widehat \Theta$ provides the optimal retention level; if $\widehat \Theta\ge 1$, then the optimal retention level is $1$, which means full reinsurance is optimal; finally, if $\widehat \Theta\le 0$, then the optimal retention level is $0$, that is no reinsurance.
Next, we  describe the sets corresponding to these three cases. Recall the definition of sets $\D_0$ and $\D_1$ in definitions \eqref{d0} and \eqref{d1}, respectively.  From \eqref{condconc}, we get that $\pd{\Psi_1}{\Theta}(\Theta, t,x,y)$ is decreasing in $\Theta \in [0,1]$, for every $(t,x,y,p) \in [0,+\infty) \times \R^3$, that is, $\pd{\Psi_1}{\Theta}(1, t,x,y,p)\le \pd{\Psi_1}{\Theta}(\Theta, t,x,y,p)\le \pd{\Psi_1}{\Theta}(0, t,x,y,p)$. We have:
\begin{itemize}
\item[(i)] if $\pd{\Psi_1}{\Theta}(0,t,x,y,p)\le 0$, then $\Theta^*(t,x,y,p)=0$, i.e. no reinsurance is chosen. This is equivalent to say that $(t,y) \in \D_0$.
\item[(ii)] if $\pd{\Psi_1}{\Theta}(1, t,x,y,p)\ge 0$, then $\Theta^*(t,x,y,p)=1$, i.e. full reinsurance is chosen. This corresponds to the case $(t,y)\in \D_1$.
\item[(iii)] the case $\pd{\Psi_1}{\Theta}(\widehat\Theta, t,x,y,p)=0$ for some $\widehat\Theta\in (0,1)$, corresponds to $(t,y)\in (\D_0\cup \D_1)^c$.
\end{itemize}

To characterize the candidate for the optimal investment portfolio $\Pi^*$, we observe that the function $\Psi_2(\Pi,t,x,y,p)$ is continuous in $\Pi$. Then, taking the first order condition we get that $\Pi^*$ given in equation \eqref{opti1} is a stationary point of the function $\Psi_2(\Pi,t,x,y,p)$, which corresponds to a maximum since the second derivative with respect to $\Pi$ is negative.

Finally, we show that the pair $(\Pi^*, \Theta^*)\in \mathcal A$, since all required integrability conditions are satisfied.  Both $\Theta^*$ and $\Pi^*$ are predictable;
moreover, for every $t \ge 0$ it holds that
\begin{align}
&\esp{\int_{0}^t\left(|\Pi_s^*||\mu(s,Y_s)| +(\Pi_s^*)^2 \sigma^2(s,Y_s)\right) \ud s} \\ &\leq c_1 \esp{\int_{0}^{t}\frac{\mu^2(s,Y_s)}{\gamma \sigma^2(s,Y_s)} \ud s}+c_2 \esp{\int_{0}^{t}h^2(s,X^H_s,Y_s)\ud s } + c_3\esp{\int_{0}^{t} \frac{\mu(s,Y_s) h(s,X^H_s,Y_s)}{\gamma \sigma(s,Y_s)} \ud s}  < \infty,
\end{align}
thanks to conditions \eqref{novikov}, \eqref{Pcoeff} and the Cauchy-Schwarz inequality, for some positive constants $c_1, c_2, c_3$. Hence, condition \eqref{int_ammiss} is satisfied.
It remains to prove that $\esp{e^{-\gamma X^{H^*}_t-P_t}}<\infty$ for each $t \ge 0$. In view of \eqref{wealthsol}, \eqref{P} and \eqref{opti1}, we have
\begin{equation}\label{eq:equality1}
\begin{split}
& \esp{e^{-\gamma X^{H^*}_t-P_t}}  = e^{-\gamma x_{0}} \mathbb E \Big{[} e^{- \frac{1}{2}\int_{0}^t \frac{\mu^2(s,Y_s)}{\sigma^2(s,Y_s)}\ud s}  e^{ - \int_{0}^t \frac{\mu(s,Y_s)}{\sigma(s,Y_s)} \ud W^S_s} e^{ - \int_{0}^t \sqrt{1-(\rho^S)^2}h(s,X^H_s,Y_s) \ud \widetilde W_s} \\
& \quad \times e^{-\frac{1}{2}\int_0^t (1-(\rho^S)^2)h^2(s,X^{H^*}_s,Y_s)\ud s}e^{\gamma \int_{0}^t \int_{I }  (1-\bar{\Theta}_{s-}) z m(\ud s,\ud z)}e^{- \int_{0}^t\lambda(s, Y_s)\int_{I } \left(e^{\gamma (1-\overline\Theta_{s-}) z}-1\right) F(s,Y_s,\ud z)\ud s}
 \Big{]},
\end{split}
\end{equation}
where $X_{0}^H=x_0 \ge 0$ is the initial wealth and $\widetilde W=\{\widetilde W_t,\ t \ge 0\}$ is an additional Brownian motion that is independent of $W^S$.
We define the process $L=\{L_t, \ t \ge 0\}$ as
\begin{align}
L_t & = e ^{-\frac{1}{2}\int_{0}^t \frac{\mu^2(r,Y_r)}{\sigma^2(r,Y_r)}\ud r - \int_{0}^t \frac{\mu(r,Y_r)}{\sigma(r,Y_r)}  \ud W_r^S -\frac{1}{2}\int_{0}^t  (1-(\rho^S)^2)h^2(r,X^{H^*}_r,Y_r)\ud r - \int_{0}^t \sqrt{1-(\rho^S)^2}h(r,X^H_r,Y_r)\ud \widetilde W_r}.
\end{align}
Then, $L$ is a square integrable martingale thanks to condition \eqref{novikov} and Assumption \ref{h_nov}. Therefore,
\begin{align}
\esp{e^{-\gamma X^{H^*}_t-P_t}} & = e^{-\gamma x_{0}} \mathbb E \Big{[} L_t e^{\gamma \int_{0}^t \int_{I }  (1-\bar{\Theta}_{s-}) z m(\ud s,\ud z)}e^{- \int_{0}^t\lambda(s, Y_s)\int_{I } \left(e^{\gamma (1-\overline\Theta_{s-}) z}-1\right) F(s,Y_s,\ud z)\ud s}
 \Big{]}\\
&\le e^{-\gamma x_{0}} \mathbb E \left[ L_t^2\right]^{1/2} \mathbb E \Big{[} e^{\gamma \int_{0}^t \int_{I }  (1-\bar{\Theta}_{s-}) z m(\ud s,\ud z)} \Big{]}^{1/2},
\end{align}
because $\int_{0}^t\lambda(s, Y_s)\int_{I } \left(e^{\gamma (1-\bar \Theta_{s-}) z}-1\right)  F(s,Y_s,\ud z)\ud s \ge 0$ $\P$-a.s. for all $t \ge 0$. Finally, we recall that  $\mathbb E \left[ L_t^2\right]<\infty$ and note that
\begin{align}
&{\mathbb{E}}\left[  e^{2\gamma \int_{0}^t \int_I (1-\bar{\Theta}_{r-}) z m(\ud r,\ud z)} \right]\leq \esp{e^{2 \gamma \sum_{i=1}^{N_t}Z_i}}\\
&= \sum_{n \ge 0}\esp{e^{2 \gamma \sum_{i=1}^{N_t}Z_i}\Big{|}N_t=n}\P(N_t=n)= \sum_{n \ge 0}\prod_{i=1}^n\esp{e^{2 \gamma Z_i}}\P(N_t=n)<\infty,
\end{align}
thanks to the assumptions on the random variables $\{Z_n\}_{n \in \mathbb{N}}$ and the fact that the process $N$ does not explode in finite time.

\subsection{Proof of Theorem \ref{thm:backward}}\label{app:thm_backward}

We notice that the optimization is taken over the set of admissible functions $\mathcal A$, even though in the backward case one would require that $\esp{e^{-\gamma X^H_T}}<\infty$ in place of $\esp{e^{-\gamma X^H_T-P_T}}<\infty$. However, because of the assumptions on model coefficients, these two conditions are equivalent.
The proof of this result uses a guess-and-verify approach.
Suppose that the value function $V(t,x,y)$ is $\mathcal C^1$ in $t$ and $\mathcal C^2$ in $x$ and $y$, then it solves the equation
\begin{align}\max_{(\Theta^B,\Pi^B) \in [0,1] \times \R} \tilde \L^HV(t,x,y)=0, \quad (t,x,y)\in [0,T)\times \R^2,\label{eq:hjb_back}\end{align}
where $ \tilde \L^H$ is the infinitesimal generator given in \eqref{mgenxy}, with the terminal condition $V(T,x,y)=-e^{-\gamma x}$, for every $x \in \R$. We guess that the value function has the form $V(t,x,y)=-e^{-\gamma x+ \phi(t,y;T)}$, where $\phi(t,y;T)$ is the unique classical solution of the problem \eqref{eq:htilde}. Plugging this expression into \eqref{eq:hjb_back} and taking the first order condition yields \eqref{eq:htilde}. The second order conditions imply that the optimal investment strategy $\Pi^{*,B}$ is given by \eqref{eq:strategia_backward} and the optimal reinsurance strategy is given by $\Theta^{*,B}(t,y)$ as in \eqref{optr}.\\
Next, we establish a verification result. Let $v(t,x,y)$ be a solution of the equation \eqref{eq:hjb_back} with the final condition $v(T,x,y)=-e^{-\gamma x}$ (that is $v(T,x,y)=V(T,x,y)$). Then, by It\^o's formula it holds that (we omit for simplicity the dependence of $X$ on the strategy $H$)
\begin{align}
	& v(T, X_T, Y_T) = v(t,x,y)+\int_t^T\tilde \L^H v(r, X_r, Y_r) \ud r \\
& \qquad + \int_t^T \Pi_r \sigma(r,Y_r) \frac{\partial v}{\partial x}(r, X_r, Y_r) \ud W^S_r + \int_t^T \beta(r,Y_r)\frac{\partial v}{\partial y}(r, X_r, Y_r) \ud W^Y_r\\
	&  \qquad + \int_t^T \int_I v(r, X_{r-}-(1-\Theta_{r-})z, Y_r)-v(r, X_{r-}, Y_r) \left(m(\ud r, \ud z)-\lambda(r, Y_r) F(r,Y_r,\ud z) \ud r\right).
\end{align}
Since $v$ satisfies equation \eqref{eq:hjb_back}, we get that
\begin{align}
	&v(T, X_T, Y_T)\leq v(t,x,y) + \int_t^T \Pi_r \sigma(r,Y_r)\frac{\partial v}{\partial x}(r, X_r, Y_r) \ud W^S_r + \int_t^T \beta(r,Y_r)\frac{\partial v}{\partial y}(r, X_r, Y_r) \ud W^Y_r \\
	& \ + \int_t^T \! \int_I v(r, X_{r-}-(1-\Theta_{r-})z, Y_r)-v(r, X_{r-}, Y_r) \left(m(\ud r, \ud z)-\lambda(r, Y_r) F(r,Y_r,\ud z) \ud r\right).\label{eq:ineq:1}
\end{align}
If the process on the right side of \eqref{eq:ineq:1} is a martingale, taking the expectation yields
$$
V(t,x,y)\leq v(t,x,y),
$$
and the equality holds if $H$ is a maximizer of equation \eqref{eq:hjb_back}.
Then, it only remains to prove that the function $V(t,x,y)=-e^{-\gamma x + \phi(t,y;T)}$ is such that
\begin{align}
	M_t=&\int_0^t \Pi_r \sigma(r,Y_r)\frac{\partial V}{\partial x}(r, X_r, Y_r) \ud W^S_r + \int_0^t \beta(r,Y_r)\frac{\partial V}{\partial y}(r, X_r, Y_r) \ud W^Y_r \\
	&+ \int_0^t \int_I V(r, X_{r-}-(1-\Theta_{r-})z, Y_r)-V(r, X_{r-}, Y_r) \left(m(\ud r, \ud z)-\lambda(r, Y_r)  F(r,Y_r,\ud z)  \ud r\right)\\
	=& \int_0^t \Pi_r \sigma(r,Y_r) \gamma e^{-\gamma X_r} e^{\phi(r, Y_r;T)} \ud W^S_r + \int_0^t \beta(r,Y_r)\pd{\phi}{y}(r, Y_r;T) e^{-\gamma X_r} e^{\phi(r, Y_r;T)} \ud W^Y_r \\
	&- \int_0^t \int_I e^{-\gamma X_r} e^{\phi(r, Y_r;T)}(e^{\gamma(1-\Theta_{r-})z}-1) \left(m(\ud r, \ud z)-\lambda(r, Y_r)  F(r,Y_r,\ud  z) \ud r\right)
\end{align}
is a martingale. To this aim, we consider the localizing sequence of random times
$$
\widetilde \tau_{n}:= \inf \Big{\{} s \in [t,T]: \ |\phi(t,Y_t;T)| > n ,\ \left|\pd{\phi}{y}(t,Y_t;T)\right| > n,\ X_t<-n  \Big{\}}, \quad n\in \mathbb N.
$$
Then, $(\widetilde \tau_n)_{ n\in \mathbb N}$ is an increasing sequence, $\lim_{n \to \infty}\tau_n \wedge T =T$ and computations similar to those in the proof of Theorem \ref{familyfeu} show that
\begin{align}
	&\mathbb{E}\bigg[\int_{0}^{T\wedge\widetilde\tau_n} \bigg{\{} \left( \Pi_r \sigma(r,Y_r) \frac{\partial V}{\partial x}(r, X_r, Y_r)\right)^2+ \left(\beta(r,Y_r)\frac{\partial V}{\partial y}(r, X_r, Y_r)\right)^2\\
	&\qquad+\int_I \left|V(r, X_{r-}-(1-\Theta_{r-})z, Y_r)-V(r, X_{r-}, Y_r) \right| \lambda(r, Y_r)  F(r,Y_r,\ud z) \bigg{\}} \ud r \bigg]<\infty.
\end{align}
This concludes the proof.


\begin{thebibliography}{24}
\providecommand{\natexlab}[1]{#1}
\providecommand{\url}[1]{\texttt{#1}}
\expandafter\ifx\csname urlstyle\endcsname\relax
  \providecommand{\doi}[1]{doi: #1}\else
  \providecommand{\doi}{doi: \begingroup \urlstyle{rm}\Url}\fi

\bibitem[Assa and Boonen(2022)]{assa2022risk}
H.~Assa and T.~J. Boonen (2022).
\newblock Risk-sharing and contingent premia in the presence of systematic
  risk: the case study of the UK Covid-19 economic losses.
\newblock In \emph{Pandemics: Insurance and Social Protection}, pages 95--126.
  Springer, Cham. \url{https://doi.org/10.1007/978-3-030-78334-1_6}

\bibitem[Brachetta and Ceci(2019)]{bec}
M.~Brachetta and C.~Ceci (2019).
\newblock Optimal proportional reinsurance and investment for stochastic factor
  models.
\newblock \emph{Insurance: Mathematics and Economics}, 87:\penalty0 15--33. \url{https://doi.org/10.1016/j.insmatheco.2019.03.006}

\bibitem[Brachetta and Schmidli(2019)]{BRASCHMIDLI}
M.~Brachetta and H.~Schmidli (2019).
\newblock Optimal reinsurance and investment in a diffusion model.
\newblock \emph{Decisions in Economics and Finance}, 43:\penalty0 341--361. \url{https://doi.org/10.1007/s10203-019-00265-8}

\bibitem[Brémaud(1981)]{BREMAUD}
P.~Brémaud.
\newblock \emph{Point Processes and Queues}.
\newblock Springer Verlag, 1981.

\bibitem[Cao et~al.(2020)Cao, Landriault, and Li]{cao2020optimal}
J.~Cao, D.~Landriault, and B.~Li (2020).
\newblock Optimal reinsurance-investment strategy for a dynamic contagion claim
  model.
\newblock \emph{Insurance: Mathematics and Economics}, 93:\penalty0 206--215. \url{https://doi.org/10.1016/j.insmatheco.2020.04.013}

\bibitem[Ceci et~al.(2022)Ceci, Colaneri, and Cretarola]{ceci2022optimal}
C.~Ceci, K.~Colaneri, and A.~Cretarola (2022).
\newblock Optimal reinsurance and investment under common shock dependence
  between financial and actuarial markets.
\newblock \emph{Insurance: Mathematics and Economics}, 105:\penalty0 252--278. \url{https://doi.org/10.1016/j.insmatheco.2022.04.011}

\bibitem[Chong(2019)]{chong2019pricing}
W.~F. Chong (2019).
\newblock Pricing and hedging equity-linked life insurance contracts beyond the
  classical paradigm: The principle of equivalent forward preferences.
\newblock \emph{Insurance: Mathematics and Economics}, 88:\penalty0 93--107. \url{https://doi.org/10.1016/j.insmatheco.2019.06.003}

\bibitem[Colaneri and Frey(2021)]{colaneri2019classical}
K.~Colaneri and R.~Frey (2021).
\newblock Classical solutions of the backward {P}{I}{D}{E} for a {M}arkov
  modulated marked point processes and applications to {C}{A}{T} bonds.
\newblock \emph{Insurance: Mathematics and Economics}, 101:\penalty0 498--507. \url{https://doi.org/10.1016/j.insmatheco.2021.09.003}

\bibitem[Colaneri et~al.(2021)Colaneri, Cretarola, and
  Salterini]{colaneri2021optimal}
K.~Colaneri, A.~Cretarola, and B.~Salterini (2021).
\newblock Optimal investment and proportional reinsurance in a regime-switching
  market model under forward preferences.
\newblock \emph{Mathematics}, 9\penalty0 (14):\penalty0 1610. \url{https://doi.org/10.3390/math9141610}

\bibitem[Delong and Gerrard(2007)]{delong2007mean}
{L}. Delong and R.~Gerrard (2007).
\newblock Mean-variance portfolio selection for a non-life insurance company.
\newblock \emph{Mathematical Methods of Operations Research}, 66\penalty0
  (2):\penalty0 339--367. \url{https://doi.org/10.1007/s00186-007-0152-2}

\bibitem[Frittelli and Maggis(2011)]{maggisfrittelli}
M.~Frittelli and M.~Maggis (2011).
\newblock Conditional certainty equivalent.
\newblock \emph{International Journal of Theoretical and Applied Finance},
  14\penalty0 (1):\penalty0 41--59. \url{https://doi.org/10.1142/S0219024911006255}

\bibitem[Gihman and Skorohod(1972)]{GIHMSKOR}
I.I. Gihman and A.V. Skorohod (1972).
\newblock \emph{Stochastic Differential Equations}.
\newblock Springer-Verlag.

\bibitem[Grandell(1991)]{grandell1991}
J.~Grandell (1991).
\newblock \emph{Aspects of Risk Theory}.
\newblock Springer {V}erlag, {N}ew {Y}ork.

\bibitem[Gu et~al.(2017)Gu, Viens, and Yi]{gu2017optimal}
A.~Gu, F.~G. Viens, and B.~Yi (2017).
\newblock Optimal reinsurance and investment strategies for insurers with
  mispricing and model ambiguity.
\newblock \emph{Insurance: Mathematics and Economics}, 72:\penalty0 235--249. \url{https://doi.org/10.1016/j.insmatheco.2016.11.007}

\bibitem[Heath and Schweizer(2000)]{HEATHSCH}
D.~Heath and M.~Schweizer (2000).
\newblock Martingales versus {PDE}s in finance: an equivalent result with
  examples.
\newblock \emph{Journal of Applied Probability}, 37(4):\penalty0 947--957. \url{https://doi.org/10.1239/jap/1014843075}

\bibitem[Irgens and Paulsen(2004)]{IRPAU}
C.~Irgens and J.~Paulsen (2004).
\newblock Optimal control of risk exposure, reinsurance and investments for
  insurance portfolios.
\newblock \emph{Insurance: Mathematics and Economics}, 35(1):\penalty0 21--51. \url{https://doi.org/10.1016/j.insmatheco.2004.04.004}

\bibitem[Liu and Ma(2009)]{liu2009optimal}
Y.~Liu and J.~Ma (2009).
\newblock Optimal reinsurance/investment problems for general insurance models.
\newblock \emph{Annals of Applied Probability}, 19\penalty0 (4):\penalty0
  1495--1528. \url{https://doi.org/10.1214/08-AAP582}

\bibitem[Musiela and Zariphopoulou(2007)]{MZ}
M.~Musiela and T.~Zariphopoulou (2007).
\newblock Investment and Valuation Under Backward and Forward Dynamic
  Exponential Utilities in a Stochastic Factor Model. In: Fu, M.C., Jarrow, R.A., Yen, JY.J., Elliott, R.J. (eds) {\em Advances in Mathematical Finance}. Applied and Numerical Harmonic Analysis. Birkh\"auser Boston, pages 303--334. \url{https://doi.org/10.1007/978-0-8176-4545-8_16}
\newblock.

\bibitem[Musiela and Zariphopoulou(2008)]{musiela2008optimal}
M.~Musiela and T.~Zariphopoulou (2008).
\newblock Optimal asset allocation under forward exponential performance
  criteria.
\newblock In \emph{Markov processes and related topics: a Festschrift for
  Thomas G. Kurtz}, pages 285--300. Institute of Mathematical Statistics. \url{https://doi.org/10.1214/074921708000000435}

\bibitem[Musiela and Zariphopoulou(2009)]{MZportchoice}
M.~Musiela and T.~Zariphopoulou (2009).
\newblock Portfolio choice under dynamic investment performance criteria.
\newblock \emph{Quantitative Finance}, 9\penalty0 (2):\penalty0 161--170. \url{https://doi.org/10.1080/14697680802624997}

\bibitem[\O{}ksendal(2013)]{OKS}
B.~\O{}ksendal (2013).
\newblock \emph{Stochastic Differential Equations: an Introduction with
  Applications}.
\newblock Springer Science \& Business Media.

\bibitem[Pham(1998)]{pham1998optimal}
H.~Pham.
\newblock Optimal stopping of controlled jump diffusion processes: a viscosity
  solution approach.
\newblock In \emph{Journal of Mathematical Systems, Estimation and Control}.
  Citeseer, 1998.

\bibitem[Pratt(1976)]{pratt1976risk}
J.~W. Pratt (1976).
\newblock Risk aversion in the small and in the large.
\newblock \emph{Econometrica}, 44\penalty0 (2):\penalty0 420--420. \url{https://doi.org/10.1016/B978-0-12-214850-7.50010-3}

\bibitem[Zariphopoulou(2001)]{ZAR_trasf}
T.~Zariphopoulou (2001).
\newblock A solution approach to valuation with unhedgeable risks.
\newblock \emph{Finance and Stochastics}, 5:\penalty0 61--82.
\url{https://doi.org/10.1007/PL00000040}
\end{thebibliography}
\end{document}